\documentclass[journal,twoside,web]{ieeecolor}
\usepackage{lcsys}
\usepackage{cite}
\usepackage{graphicx}
\usepackage{textcomp}

\usepackage[dvipsnames]{xcolor}
\usepackage{tikz}
\usepackage[english]{babel}
\usepackage[ruled]{algorithm2e}
\usepackage{algcompatible}
\SetKwRepeat{Do}{do}{while}

\usepackage{amsmath,amssymb,amsfonts,amsthm}
\newtheorem{remark}{Remark}[section]
\newtheorem{assumption}{Assumption}[section]

\newtheorem{theorem}{Theorem}[section]

\newtheorem{corollary}{Corollary}[section]
\usepackage{arydshln}
\usepackage{float}
\usepackage{empheq}

\def\BibTeX{{\rm B\kern-.05em{\sc i\kern-.025em b}\kern-.08em
		T\kern-.1667em\lower.7ex\hbox{E}\kern-.125emX}}
\begin{document}
	
	\title{\LARGE Set-based and Dynamical Feedback-augmented Hands-off Control\\\color{white} (Andrei's Version)\vspace{-3mm}}
	\author{Andrei Speril\u{a}, Sorin Olaru, and St\'{e}phane Drobot\vspace{-8mm}
		\thanks{Andrei Speril\u{a} and Sorin Olaru are with the Laboratoire des Signaux et Syst\` emes, CentraleSup\'elec, UPSaclay, Gif-sur-Yvette, France,
			{\tt\footnotesize \{andrei.sperila,sorin.olaru\}@centralesupelec.fr}}%
		\thanks{St\'{e}phane Drobot is with the R\'{e}seau de Transport d'\'{E}lectricit\'{e}, Paris, France,
			{\tt\footnotesize stephane.drobot@rte-france.com}}
	}
	
	\maketitle
	\thispagestyle{empty}
	
	\begin{abstract}
		A novel set-theoretical approach to hands-off control is proposed, focusing on spatial arguments for command limitation rather than temporal ones. By employing dynamical feedback alongside invariant set-based constraints, actuation is employed only to drive the system's state within a ``hands-off region'' of its state-space, where the plant can freely evolve in open-loop configuration. A computationally-efficient procedure with strong theoretical guarantees is devised, and its effectiveness is showcased via an intuitive practical example.\vspace{-1mm}
	\end{abstract}
	
	\begin{IEEEkeywords}
		command switching, dynamical feedback, hands-off control, invariant sets
	\end{IEEEkeywords}\vspace{-3mm}
	
	\section{Introduction}\label{sec:intro}
	
	\IEEEPARstart{C}{lassical} formulations of the hands-off problem, such as the seminal work presented in Chapter 6 of \cite{ho_book}, focused on minimizing the amount of time in which control action is applied to a system, while ensuring satisfactory closed-loop performance guarantees in the presence of disturbance or model uncertainty. This approach was extensively investigated in \cite{ho_art}, where all the aforementioned aspects were addressed for a broad class of systems. While the problem of \emph{time-optimal} hands-off control is still being studied in system-theoretical literature \cite{ho_aux1,ho_aux2,ho_aux3}, the spatial aspect of this design problem has not received the same level of attention.
	
	In this paper, we tackle the problem of computing and then enforcing, by means of dynamical feedback, the existence of a \emph{hands-off region} in the state space of a finite-dimensional, linear, and time-invariant (FDLTI) system. The key feature of such a region is the fact that, when the system's state is located within it, the proposed control algorithm does not act upon the plant. It is only when the system's state leaves the hands-off region that the controller is allowed to act, in order to drive the system's state back into the aforementioned region. Once this is achieved, the feedback loop is interrupted, and the system's state is allowed to drift within the hands-off region.
	
	Despite the shared switching mechanism, the technique presented in this paper is notably distinct from the \emph{safety filter}-based approach proposed in \cite{safety_filt}, both in terms of the commutation's nature (our approach switches between open- and closed-loop configurations, rather than two distinct control laws) and in terms of the associated computational cost (our solution involves merely the implementation of an FDLTI controller). Likewise, due to the set-theoretical focus of our approach, the use of strong set invariance (see, for example, Chapter 4 of \cite{STMC}) and that of inexpensive state-space-based control laws differentiates our proposed solution from the \emph{barrier function}-based technique discussed in \cite{barrier_func}.

	Once notation is introduced in Section~\ref{subsec:denote}, the control problem described above will be rigorously formulated in Section~\ref{subsec:prob_st}. Following this, the control solution, which ensures the desired functioning, will be presented in Section~\ref{sec:main_res}, and its application will be showcased in Section~\ref{sec:num_ex}. Alongside Section~\ref{sec:outro}, which contains a set of concluding remarks, the manuscript also includes an appendix, which holds the proofs of the main theoretical results presented in this paper.\vspace{-1mm}
	
	\section{Preliminaries}\label{sec:prelims}\vspace{-1mm}
	
	\subsection{Notation}\label{subsec:denote}\vspace{-1mm}

	Let $\mathbb{N}$ and $\mathbb{R}$ denote the set of natural and real numbers, respectively. Additionally, let $\mathbb{N}_{>0}:=\{n\in\mathbb{N}:n>0\}$ and, for $a,b\in\mathbb{R}$ with $a\leq b$, let $\mathbb{N}_{[a,b]}:=\{n\in\mathbb{N}:a\leq n\leq b\}$. 
	
	For any set $\mathcal{M}$, let $\mathcal{M}^{n}$ and $\mathcal{M}^{p\times m}$ stand for the set of all vectors of dimension $n$ and, respectively, the set of all matrices of dimension $p\times m$ whose entries belong to $\mathcal{M}$. The operator $\|\cdot\|$ represents the (induced, for matrix arguments) 2-norm and $e_i$ stands for the $i^\text{th}$ column of the identity matrix, whose dimension is inferred from the available context. For any $n\in\mathbb{N}_{>0}$, we proceed to denote by $0_n$ and by $1_n$ those\newline $n$-dimensional vectors whose entries are all $0$ and, respectively, all $1$. For any two sets $\mathcal{X},\mathcal{Y}\subseteq\mathbb{R}^n$, given an arbitrary $n\in\mathbb{N}_{>0}$, the operation $\mathcal{X}\oplus\mathcal{Y}$ denotes the Minkowski sum, while $\mathcal{X}\ominus\mathcal{Y}$ stands for the Pontryagin difference. Finally, we employ the shorthand notation $(-\mathcal{X}):=\{-x:x\in\mathcal{X}\}$, and for any two polyhedra (see, for example, Section 3.3 in~\cite{STMC}) $\mathcal{Y}_1,\mathcal{Y}_2\subseteq\mathbb{R}^n$, the operator $\mathrm{conv}(\mathcal{Y}_1,\mathcal{Y}_2)$ computes the \emph{convex hull} of the vertices which make up $\mathcal{Y}_1$ and $\mathcal{Y}_2$.\vspace{-1mm}

	\subsection{Problem Statement}\label{subsec:prob_st}\vspace{-1mm}
	
	We consider a state-space system with an $n$-dimensional state vector, denoted by $x[k]$, whose dynamics are as follows\vspace{-1mm}
	\begin{equation}\label{eq:ss_sys}
		\left\{\begin{aligned}
			x[k+1]=&\ Ax[k]+\sigma[k]Bu[k]+(1-\sigma[k])d[k]+w[k],\\
			y[k]=&\ x[k]+v[k],
		\end{aligned}\right.\vspace{-1mm}
	\end{equation}
	The $m$-dimensional vector $u[k]$ represents the system's controlled inputs, whereas $d[k]\in\mathcal{D}$ stands for the uncontrolled ones. The vector $w[k]\in\mathcal{W}$ represents the process noise, while $\sigma[k]\in\mathbb{B}:=\{0,1\}$ is a \emph{switching signal}
	which governs the system's binary functioning, and $v[k]\in\mathcal{V}$ designates the measurement noise. We also make the following assumption.\vspace{-1mm}
	
	\begin{assumption}\label{asu:sets}
		All the sets discussed in this section (with the exception of $\mathbb{B}$) are polytopic (see Section 3.3 in \cite{STMC}) and they include the zero vector inside their respective interiors.\vspace{-1mm}
	\end{assumption}
	
	Given $\mathcal{U}\subset\mathbb{R}^m$ along with $\mathcal{S}\subset\mathcal{X}\subset\mathbb{R}^n$ which satisfy\vspace{-1mm}
	\begin{subequations}
		\begin{align}\label{eq:incl_a}
			&\exists\, \varepsilon_p,\varepsilon_m\in(0,1)\text{ s.t. }
			\mathcal{V}\subseteq\varepsilon_p\,\mathcal{S}\text{ and }(-\mathcal{V})\subseteq\varepsilon_m\,\mathcal{S},\\
			\label{eq:incl_b}
			&\exists\, \delta>0\text{ s.t. }
			\overline{\mathcal{B}}_\delta:=\{r\in\mathbb{R}^n:\|r\|\leq\delta\}\subseteq\mathcal{S},\\
			&\mathcal{S}^+:=A\,\mathcal{S}\oplus\mathcal{W}\oplus\mathcal{D}\subseteq\mathcal{X},\label{eq:incl_c}
		\end{align}
	\end{subequations}
	we assume that the state is initialized somewhere in the set $\mathcal{S}_c:=\mathrm{conv}(\mathcal{S},\mathcal{S}^+)\subseteq\mathcal{X}$, with the inclusion following from Assumption~\ref{asu:sets}, from $\mathcal{S}\subset\mathcal{X}$ and from \eqref{eq:incl_c}. The goal of the paper is to design a control scheme that computes appropriate time-domain functions $u[k]$ and $\sigma[k]$ such that the system described by \eqref{eq:ss_sys} functions as follows:
	\begin{enumerate}
		\item[$i)$] the constraints $x[k]\in\mathcal{X}$ along with $u[k]\in\mathcal{U}$ must be satisfied at all time instants;\smallskip
		
		\item[$ii)$] If the system has been initialized outside of $\mathcal{S}$ or if the system's state strays outside of $\mathcal{S}$, then $\sigma[k]$ must be set to $1$ and the state vector must be driven to within some $\Omega_I\subseteq\mathcal{S}$ (to be designed in the sequel);\smallskip
		
		\item[$iii)$] After the system's state has been brought inside $\Omega_I$ or if the system has been initialized in $\mathcal{S}$, then $\sigma[k]$ must be set to $0$ and the system must be allowed to drift within $\mathcal{S}$, as influenced by $w[k]$ and by $d[k]$.
	\end{enumerate}
	
	\begin{remark}\label{rem:measurement}
		Recall the fact that $x[k]$ is not directly available for measurement, due to the noise vector $v[k]$. Checking, therefore, whether $x[k]\in\mathcal{S}$ must be done via the sufficient condition $y[k]\in(\mathcal{S}\ominus(-\mathcal{V}))$ (which is non-empty, due to \eqref{eq:incl_a}-\eqref{eq:incl_b} and to Assumption~\ref{asu:sets}). Consequently, we will treat $y[k]\not\in(\mathcal{S}\ominus(-\mathcal{V}))$ as corresponding to $x[k]\not\in\mathcal{S}$, since we are unable to guarantee the fact that $x[k]\in\mathcal{S}$. Moreover, the same arguments hold when assessing whether $x[k]\in\Omega_I$, which will be designed such that $(\Omega_I\ominus(-\mathcal{V}))\neq\{\emptyset\}$. The use of certain observer-based strategies could mitigate these shortcomings, yet their use is beyond the scope of this paper, and we assume that $\mathcal{V}$ is sufficiently small (with respect to $\mathcal{S}$ and $\Omega_I$) so as not to impact performance significantly.
	\end{remark}
	
	Having presented the desired functioning of the system described by \eqref{eq:ss_sys}, we now propose a control solution which satisfies the operating principles laid out in points $i)$-$iii)$.

	\section{Main Results}\label{sec:main_res}
	
	\subsection{Theoretical Aspects}
	
	One of the main contributions of the work presented in this paper is the fact that, when $\sigma[k]=1$, we compute $u[k]$ via dynamical feedback. Thus, consider the systems
	\begin{equation}\label{eq:ss_ctl}
		\left\{\begin{aligned}
			x_K[k+1] = A_K x_K[k] + B_K y[k],\\
			u[k] = C_K x_K[k] + D_K y[k],
		\end{aligned}
		\right.
	\end{equation}
	in which $x_K[k]$ represents the controller's $n_K$-dimensional state vector. We denote by $\xi[k]:=\begin{bmatrix}
		x^\top[k] & x_K^\top[k]
	\end{bmatrix}^\top$ the closed-loop state vector of the system described by \eqref{eq:cl_dyn}, along with the exogenous signal vector $z[k]:=\begin{bmatrix}
		w^\top[k] & v^\top[k]
	\end{bmatrix}^\top$, which satisfies $z[k]\in\mathcal{Z}:=\mathcal{W}\times\mathcal{V}$. By employing this notation, the closed-loop dynamics can be expressed as
	\begin{equation}\label{eq:cl_dyn}
		\left\{\begin{aligned}
			\xi[k+1] =&\ A_{CL} \xi[k] + B_{CL} z[k],\\
			u[k] =&\ C_{CL} \xi[k] + D_{CL} z[k],
		\end{aligned}\right.
	\end{equation}
	where
	\begin{equation}\label{eq:cl_real}
		\left\{
		\begin{array}{ll}
			\hspace{-2mm}A_{CL}\hspace{-0.25mm}=\hspace{-0.25mm}\begin{bmatrix}
				A+BD_K & \hspace{-2mm} BC_K \\ B_K & \hspace{-2mm} A_K
			\end{bmatrix},&\hspace{-2mm}
			B_{CL}\hspace{-0.25mm}=\hspace{-0.25mm}\begin{bmatrix}
				I_n & \hspace{-2mm} BD_K \\ O & \hspace{-2mm} B_K
			\end{bmatrix},\\\vspace{-2mm}\\
			\hspace{-2mm}C_{CL}\hspace{-0.25mm}=\hspace{-0.25mm}\begin{bmatrix}
				D_K & C_K
			\end{bmatrix},&\hspace{-2mm}
			D_{CL}\hspace{-0.25mm}=\hspace{-0.25mm}\begin{bmatrix}
				O & D_K
			\end{bmatrix}.
		\end{array}
		\right.
	\end{equation}
	
	To ensure that the system given by \eqref{eq:ss_sys} obeys the functioning principles laid out in points $i)$-$iii)$ from Section~\ref{subsec:prob_st}, we adopt a set-theoretical formulation inspired by the work from \cite{io_sets} for our control problem. More precisely, we will construct a pair of sets $\Omega_I\subset\mathbb{R}^n$ and $\Omega_O\subset\mathbb{R}^{n+n_K}$, which are (in effect) subsets of the state-spaces belonging to the dynamical systems described in \eqref{eq:ss_sys} and \eqref{eq:cl_dyn}, respectively. We introduce the set $\mathcal{N}:=(-\mathcal{V})\oplus\mathcal{V}$ and we assume that $0_n\in\Omega_I$, that $0_{n+n_K}\in\Omega_O$, that both $\Omega_I$ and $\Omega_O$ are polyhedral in nature, and these two sets satisfy:
	
	\begin{enumerate}
		
		\item[I1)] $\Omega_I\subseteq\varepsilon_s\,\mathcal{S}\text{, for a given }\varepsilon_s\in(0,1)\,;$\smallskip
		
		\item[I2)] There exists $\alpha>0$ such that\vspace{-1mm}
		\begin{equation*}
			\overline{\mathcal{B}}_\alpha:=\{r\in\mathbb{R}^n:\|r\|\leq\alpha\}\subseteq(\Omega_I\ominus\mathcal{N})\,;\vspace{-1mm}
		\end{equation*}
		
		\item[I3)] There exists $\beta\in(0,1)$ such that\vspace{-1mm}
		\begin{equation*}
			\textstyle\bigoplus_{i=0}^N\begin{bmatrix}
				I_n \ \, O
			\end{bmatrix}A_{CL}^iB_{CL}^{ }\mathcal{Z}\subseteq\beta\,(\Omega_I\ominus\mathcal{N}),\forall N\in\mathbb{N}\,;
		\end{equation*}
		
		\item[O1)] $\mathcal{S}_c\times\{0_{n_K}\}\subseteq\Omega_O$\,;\smallskip
		
		\item[O2)] $A_{CL}\,{\Omega}_O\oplus B_{CL}\,\mathcal{Z}\subseteq{\Omega}_O$\,;\smallskip
		
		\item[O3)] $\begin{bmatrix}
			I_n & O
		\end{bmatrix}\Omega_O\subseteq\mathcal{X}$\,;\smallskip
		
		\item[O4)] $C_{CL}\,\Omega_O\oplus D_{CL}\,\mathcal{Z}\subseteq \mathcal{U}$\,.\vspace{-1mm}
		
	\end{enumerate}

	\begin{remark}\label{rem:eps}
		Note that $\varepsilon_s\neq 1$ in I1) and, thus, the frontiers of $\Omega_I$ and $\mathcal{S}$ may never intersect (presuming that $\mathcal{S}\neq\{0\}$ and recalling Assumption~\ref{asu:sets}). If this were not the case, then certain edge cases could be constructed where an arbitrarily small $d[k]$ (in terms of its norm value) would be able to propel the state outside of $\mathcal{S}$. In order to prevent such rapid oscillations between the system's two modes of functioning, the scalar $\varepsilon_s$ will be treated as a design parameter.\vspace{-1mm}
	\end{remark}
	
	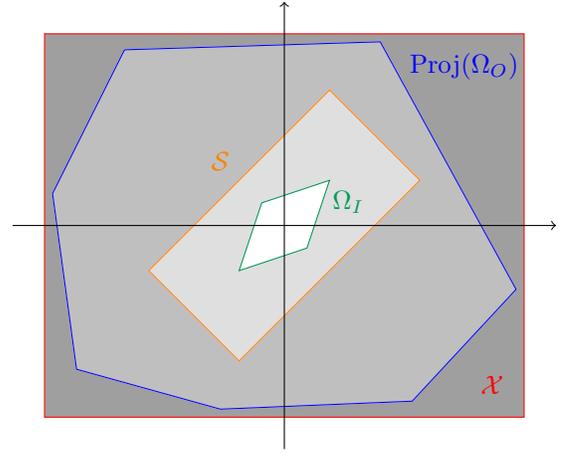
\begin{figure}[t]
		\centering
		\begin{tikzpicture}[scale=0.425]
			\draw  (-7.5,-6) [color=red, fill=gray, fill opacity=0.75] rectangle (7.5,6);
			\draw [draw=white, fill=white, fill opacity=1] (-5,5.5) -- (-7.25,1) -- (-6.5,-4.5) -- (-2,-5.75) -- (4, -5.5) -- (7.25,-2) -- (3,5.75) -- (-5,5.5);
			\draw [draw=blue, fill=gray, fill opacity=0.5] (-5,5.5) -- (-7.25,1) -- (-6.5,-4.5) -- (-2,-5.75) -- (4, -5.5) -- (7.25,-2) -- (3,5.75) -- (-5,5.5);
			\draw [color=white, rotate around={315:(0,0)}, fill=white, fill opacity=1] (-2,4) -- (2,4) -- (2,-4) -- (-2,-4) -- cycle;
			\draw [color=orange, rotate around={315:(0,0)}, fill=gray, fill opacity=0.25] (-2,4) -- (2,4) -- (2,-4) -- (-2,-4) -- cycle;
			\draw [color=ForestGreen, rotate around={45:(0, 0)}, fill=white, fill opacity=1] (-2,0) -- (0,1) -- (2,0) -- (0,-1) -- cycle;
			\draw [->] (0,-7) -- (0,7);
			\draw [->] (-8.5,0) -- (8.5,0);
			\node at (5.625,5) {\color{blue}$\mathrm{Proj}(\Omega_O)$};
			\node at (6.5,-5) {\color{red}$\mathcal{X}$};
			\node at (2,0.75) {\color{ForestGreen}$\Omega_I$};
			\node at (-2,2) {\color{orange}$\mathcal{S}$};
		\end{tikzpicture}
		\caption{The greyscale shading inside the sets indicates the connection between states and their corresponding control action: dark grey indicates admissible states that will never be reached (due to the applied commands), medium grey marks full ``hands-on control'' space, light grey designates the transitional space, and white is assigned to full ``hands-off control'' space.\vspace{-5mm}}
		\label{fig:set_incl}
	\end{figure}
	
	To get a better sense of the most important sets which are employed in our control system, we illustrate via Figure~\ref{fig:set_incl} the concept for a system of type \eqref{eq:ss_sys}, in which we have $n=2$. These schematic representations provide a deeper insight into the geometric properties of the aforementioned sets and, in addition to this, we also highlight, via the following result, the way in which these properties come into play, with the aim of ensuring the desired functioning for the system in \eqref{eq:ss_sys}.

	\begin{algorithm}[t]
		\textbf{Initialization:} Read $y[k_0]$ from the sensors\; \eIf{$y[k_0]\in(\mathcal{S}\ominus(-\mathcal{V}))$}{Set $\sigma[k_0]\gets 0$ and $u[k_0]\gets 0_m$, do $k\gets k_0+1$, then go to \textbf{Monitoring}\;}{Set $k\gets k_0$, $\sigma[k]\gets 1$ and $x_K[k]\gets 0_{n_K}$, then go to \textbf{Control}\;}
		\textbf{Monitoring:} Read $y[k]$ from the sensors\; \eIf{$y[k]\in(\mathcal{S}\ominus(-\mathcal{V}))$}{\eIf{$\sigma[k-1]=1$\emph{ \bf and }$y[k]\not\in(\Omega_I\ominus(-\mathcal{V}))$}{Set $\sigma[k]\gets 1$, then go to \textbf{Control}\;}{Set $\sigma[k]\gets 0$ and $u[k]\gets 0_m$, do $k\gets k+1$, then go to \textbf{Monitoring}\;}}{\If{$\sigma[k-1]=0$}{Set $x_K[k]\gets 0_{n_K}$\;}Set $\sigma[k]\gets 1$, then go to \textbf{Control}\;}
		\textbf{Control:} Compute $x_K[k+1]$ and $u[k]$ as in \eqref{eq:cl_dyn}, do $k\gets k+1$, then go to \textbf{Monitoring}\;
		\caption{Control procedure for type \eqref{eq:ss_sys} systems.}
		\label{alg:ho_periodic}
	\end{algorithm}
	
	\begin{theorem}\label{thm:algo}
		Consider a system of type \eqref{eq:ss_sys} and a controller of type \eqref{eq:ss_ctl}, for which $A_{CL}$ is a Schur matrix and for which $\Omega_I$ along with $\Omega_O$ satisfy I1)-I3) and O1)-O4), respectively. If the system's initialization satisfies $x[k_0]\in\mathcal{S}_c\,$, then applying Algorithm~\ref{alg:ho_periodic} ensures that:\smallskip
		\begin{enumerate}
			\item[a)] $x[k]\in\mathcal{X}$ and $u[k]\in\mathcal{U}$, for all $k\geq k_0;$\smallskip
			
			\item[b)] There exists $T_{\mathrm{max}}\in\mathbb{N}_{>0}$ which satisfies
			\begin{equation*}
				\left\{
				\begin{aligned}
					&\sigma[k_0]=1\Rightarrow\exists\, T_0\in\mathbb{N}_{[1,T_{\mathrm{max}}]}\emph{ s.t. }\sigma[k_0+T_0]=0,\\
					&k>k_0,\,\sigma[k]=1,\,\sigma[k-1]=0\Rightarrow\\
					&\qquad\qquad\quad\Rightarrow\exists\, T\in\mathbb{N}_{[1,T_{\mathrm{max}}]}\emph{ s.t. }\sigma[k+T]=0;
				\end{aligned}
				\right.
			\end{equation*}
			
			\item[c)] $\sigma[k]=0\Rightarrow x[k]\in\mathcal{S}$, for all $k\geq k_0.$
		\end{enumerate}
	\end{theorem}
	\begin{proof}
		See the appendix.
	\end{proof}
	
	The following consequence of Theorem~\ref{thm:algo} plays a crucial role in obtaining 
	computationally inexpensive implementations for the controllers given in \eqref{eq:ss_ctl}.
	
	\begin{corollary}\label{cor:any_real}
		Let $\mathbf{K}(z)\hspace{-0.5mm}:=\hspace{-0.5mm}C_K(zI_{n_K}-A_K)^{-1}B_K\hspace{-0.5mm}+\hspace{-0.5mm}D_K$ represent the transfer function matrix of the controller mentioned in Theorem~\ref{thm:algo} (which makes $A_{CL}$ a Schur matrix). Then, the latter result still holds when employing \emph{any other realization} of $\mathbf{K}(z)$ in the \textbf{\emph{Control}} step of Algorithm~\ref{alg:ho_periodic}.
	\end{corollary}
	
	\begin{proof}
		The result follows from the fact that, when the controller's state is initialized by the zero vector, the command signals depend only upon the Markov parameters (see Section 6.2.1 in \cite{Kai}) of $\mathbf{K}(z)$, which are invariant w.r.t. all of its realizations (see the proof of Theorem~6.2-3 in \cite{Kai}).
	\end{proof}
	
	\begin{remark}
		By leveraging the result from Corollary~\ref{cor:any_real}, notice that we may always obtain a computationally efficient implementation for the controller described by \eqref{eq:ss_ctl}. More precisely, this boils down to expressing a \emph{minimal} (see, for example, Chapter~3 in \cite{zhou}) realization for the aforementioned controller and then employing this representation to compute $u[k]$, rather than the generic one used previously.\vspace{-1mm}
	\end{remark}
	
	Although conditions I1)-I3) and O1)-O4) ensure (as shown in Theorem~\ref{thm:algo}) the desired hands-off behaviour, notice that the conditions related to $\Omega_I$ are significantly more involved than the ones concerning $\Omega_O$. In order to remedy this fact, we proceed to lift the former set into a higher dimension and state the following result, which provides more tractable formulations for the conditions stated in points I1)-I3).\vspace{-1mm}
	
	\begin{theorem}\label{thm:set_design}
		Consider the closed-loop dynamics given in \eqref{eq:cl_dyn}-\eqref{eq:cl_real} along with a scalar $\eta\in(0,1)$ and a set\vspace{-1mm}
		\begin{equation}\label{eq:OI_lifted}
			\hspace{-2mm}\small\begin{array}{l}
				\widetilde{\Omega}_I\hspace{-0.25mm}:=\hspace{-0.25mm}\left\{\xi\hspace{-0.5mm}\in\hspace{-0.5mm}\mathbb{R}^{n+n_K}\hspace{-0.5mm}:\hspace{-0.5mm}\small\begin{bmatrix}
					\phantom{-}\widetilde{H}_I \\ -\widetilde{H}_I
				\end{bmatrix}\hspace{-0.25mm}\xi\hspace{-0.25mm}\leq\hspace{-0.5mm} \small\begin{bmatrix}
					1_{n+n_K} \\ \vspace{-3mm} \\ 1_{n+n_K}
				\end{bmatrix}\hspace{-0.5mm},\hspace{-0.25mm}\,\hspace{-0.25mm}\det\hspace{-0.25mm}\left(\widetilde{H}_I \right)\hspace{-0.25mm}\neq\hspace{-0.25mm} 0\hspace{-0.25mm}\right\}\hspace{-0.25mm}.
			\end{array}\normalsize\hspace{-2mm}\vspace{-1mm}
		\end{equation}
		Assume that $\varepsilon_p+\varepsilon_m<\varepsilon_s<1$ and that the following two conditions hold:
		\begin{enumerate}
			\item[C1)] $A_{CL}\,\widetilde{\Omega}_I\oplus B_{CL}\,\mathcal{Z}\subseteq\eta\,\widetilde{\Omega}_I$;\smallskip
			
			\item[C2)] $\begin{bmatrix}
				I_n & O
			\end{bmatrix}\widetilde{\Omega}_I\subseteq(\varepsilon_s-\varepsilon_p-\varepsilon_m)\,\mathcal{S}$.\smallskip
		\end{enumerate}
		Then, we have that:
		\begin{enumerate}
			\item[a)] The set $\Omega_I = \begin{bmatrix}
				I_n & O
			\end{bmatrix}\widetilde{\Omega}_I\oplus\mathcal{N}$ satisfies I1)-I3) with $\alpha = \frac{1}{\left\|\widetilde{H}_I\right\|}$ and $\beta = \eta$\,;\smallskip
			
			\item[b)] The matrix $A_{CL}$ is Schur.\vspace{-1mm} 
		\end{enumerate}
	\end{theorem}
	\begin{proof}
		See the appendix.\vspace{-1mm}
	\end{proof}
	
	The conditions stated in points C1)-C2) are of the same type as those in O1)-O4) and, most importantly, they allow for the joint design of the controller in tandem with the pair of sets $\Omega_I$ and $\Omega_O$. The means by which this may be achieved represents the topic of the paper's next section.\vspace{-1mm}
	
	\begin{algorithm}[t]
		\textbf{Input data:} The dynamics of type \eqref{eq:ss_sys}, its associated sets from Section~\ref{subsec:prob_st}, and the scalars $n_K\geq0$, $0<\eta<1$, $\varepsilon_p+\varepsilon_m<\varepsilon_s<1$ and $0<\varepsilon_{c}\ll 1$\;\smallskip
		
		\textbf{Initialization:} Define $M_0:=\small\begin{bmatrix}
			\widetilde{H}_I^0 & \widetilde{H}_{Ii}^0 & {H}_O^0 & {H}_{Oi}^0
		\end{bmatrix}$,
		$J^0:=\mathrm{tr}\left(M_0^{}M_0^\top\right)$ and $Q^0:=O$, then choose invertible matrices $Y_{I1j}^0$, $Y_{I2j}^0$, $Y_{O1j}^0$, $Y_{O2j}^0$ and $Y_\mathcal{U\ell}^0$ for all $j\in\mathbb{N}_{[1,n+n_K]}$ and all $\ell\in\mathbb{N}_{[1,m]}$ such that $\mathcal{P}^0(Q^0)$ in \eqref{eq:opt_a}-\eqref{eq:opt_s} is feasible, and solve $\mathcal{P}^0(Q^0)$\;\smallskip
		
		\textbf{Iteration setup:} Assign $k\gets 0$ and employ an optimizer of $\mathcal{P}^0\left(Q^0\right)$ to form the matrix $G^0\gets\begin{bmatrix}
			\widetilde{H}_{I}^{0}\widetilde{H}_{Ii}^{0}-I_{n+n_K} & {H}_{O}^{0}{H}_{Oi}^{0}-I_{n+n_K}
		\end{bmatrix}$\;
		
		\Do{$\mathrm{tr}\left(G^{k-1}(G^{k-1})^\top\right)-\mathrm{tr}\left(G^k(G^k)^\top\right)\geq\varepsilon_c\,$}{
			$k\gets k+1$\;
			\For{$j\in\mathbb{N}_{[1,n+n_K]}$}{
				$Y_{I1j}^k\gets \left(X_{I1j}^{k-1}\right)^{-1}\widetilde{H}_{Ii}^{k-1}$\;\smallskip
				$Y_{I2j}^k\gets \left(X_{I2j}^{k-1}\right)^{-1}\left(\widetilde{H}_{I}^{k-1}\right)^\top$\;\smallskip
				$Y_{O1j}^k\gets \left(X_{O1j}^{k-1}\right)^{-1}H_{Oi}^{k-1}$\;\smallskip
				$Y_{O2j}^k\gets \left(X_{O2j}^{k-1}\right)^{-1}\left(H_O^{k-1}\right)^\top$\;
			}
			\For{$\ell\in\mathbb{N}_{[1,n_\mathcal{U}]}$}{
				$Y_{\mathcal{U}\ell}^k\gets \left(X_{\mathcal{U}\ell}^{k-1}\right)^{-1}\left(H_O^{k-1}\right)^\top$\;
			}
			\eIf{$k\mod 2>0$}{$Q^k:=\begin{bmatrix}
					\widetilde{H}_{I}^{k}-\widetilde{H}_{I}^{k-1} & {H}_{O}^{k}-{H}_{O}^{k-1}
				\end{bmatrix}$\;}{$Q^k:=\begin{bmatrix}
					\widetilde{H}_{Ii}^{k}-\widetilde{H}_{Ii}^{k-1} & {H}_{Oi}^{k}-{H}_{Oi}^{k-1}
				\end{bmatrix}$\;}
			$G^k:=\begin{bmatrix}
				\widetilde{H}_{I}^{k}\widetilde{H}_{Ii}^{k}-I_{n+n_K} & {H}_{O}^{k}{H}_{Oi}^{k}-I_{n+n_K}
			\end{bmatrix}$\;\smallskip
			$J^k:=\mathrm{tr}\left(G^k(G^k)^\top\right)$\;\smallskip
			Solve $\mathcal{P}^k\left(Q^k\right)$ and, for its optimizer, evaluate the expression of $G^k$, while storing its value within a variable of the same name\;\smallskip
		}
		\caption{Joint synthesis procedure for the plant's dynamical controller and for the pair of sets.}
		\label{alg:syn_proc}
	\end{algorithm}
	
	\subsection{Synthesis Procedure}\label{subsec:syn_proc}\vspace{-1mm}
	
	One of the key challenges in designing the dynamical system given in \eqref{eq:ss_ctl} simultaneously with a pair of sets which (individually) satisfy C1)-C2) along with O1)-O4) is given by the inherent nonlinearity of the resulting expressions, with respect to the employed design parameters. A particularly suitable solution to overcome this hurdle is the S-procedure-based approach proposed in \cite{svi1} and \cite{svi2}, which notably employs the so-called \emph{slack variable identity} to bypass the aforementioned nonlinearity. To this end, we proceed to select the vectors $\omega_i\in\mathbb{R}^n$, for $i\in\mathbb{N}_{[1,n_c]}$, and a quadruplet of matrices $H_{\mathcal{S}}\in\mathbb{R}^{n_\mathcal{S}\times n}$, $H_{\mathcal{X}}\in\mathbb{R}^{n_\mathcal{X}\times n}$, $H_{\mathcal{U}}\in\mathbb{R}^{n_\mathcal{U}\times m}$ and $H_{\mathcal{Z}}\in\mathbb{R}^{n_{\mathcal{Z}}\times 2n}$ such that\vspace{-1mm}
	\begin{equation}\label{eq:set_H_reps}
		\left\{\begin{aligned}
			\mathcal{S}=&\left\{x\in\mathbb{R}^n:H_\mathcal{S}x\leq1_{n_\mathcal{S}}\right\},\\
			\mathcal{S}_c=&\left\{x\in\mathbb{R}^n:x=\textstyle\sum_{i=1}^{n_c}\alpha_i\omega_i,\alpha_i\geq 0,1=\textstyle\sum_{i=1}^{n_c}\alpha_i\right\},\\
			\mathcal{X}=&\left\{x\in\mathbb{R}^n:H_\mathcal{X}x\leq1_{n_\mathcal{X}}\right\},\\
			\,\mathcal{U}=&\left\{u\in\mathbb{R}^m:H_\mathcal{U}u\leq1_{n_\mathcal{U}}\right\},\\
			\mathcal{Z}\subseteq&\left\{z\in\mathbb{R}^{2n}:-1_{n_{\mathcal{Z}}}\leq H_\mathcal{Z}z\leq 1_{n_{\mathcal{Z}}}\right\}=:\mathcal{Z}_s\,,\\
		\end{aligned}\right.
	\end{equation}
	and we proceed to state the following result, which formalizes the theoretical benefits of our synthesis procedure.
	
	\begin{figure*}
		\begin{subequations}
			\raisebox{-85mm}[0pt][0pt]{$\mathcal{P}^k\left(Q^k\right):=\left\{\begin{array}{l}
					\vspace{160mm}
				\end{array}\right.$}
			\begin{equation}
				\label{eq:opt_a}
				\hspace{5mm}\min\limits_{A_K^k,B_K^k,C_K^k,D_K^k,\widetilde{H}_{I}^{k},\widetilde{H}_{Ii}^{k},{H}_{O}^{k},{H}_{Oi}^{k},X_{I1j}^k,X_{I2j}^k,X_{O1j}^k,X_{O2j}^k,X_{\mathcal{U}\ell}^k,D_{I1j}^k,D_{I2j}^k,D_{O1j}^k,D_{O2j}^k,D_{\mathcal{S}q}^k,D_{\mathcal{X}p}^k,D_{\mathcal{U}1\ell}^k,D_{\mathcal{U}2\ell}^k}J^k\vspace{-2mm}
			\end{equation}
			\begin{empheq}[left=\hspace{20mm}\text{subject to}\empheqlbrace]{align}
				& D_{I1j}^k,D_{I2j}^k,D_{O1j}^k,D_{O2j}^k\succeq O,D_{I1j}^k,D_{I2j}^k,D_{O1j}^k,D_{O2j}^k\text{ diagonal},\,\forall\,j\in\mathbb{N}_{[1,n+n_K]},\label{eq:opt_b}\\
				& D_{\mathcal{S}q}^k\succeq O, D_{\mathcal{S}q}^k\text{ diagonal},\,\forall\,q\in\mathbb{N}_{[1,n_\mathcal{S}]},\label{eq:opt_c}\\
				& D_{\mathcal{X}p}^k\succeq O, D_{\mathcal{X}p}^k\text{ diagonal},\,\forall\,p\in\mathbb{N}_{[1,n_\mathcal{X}]},\label{eq:opt_d}\\
				& D_{\mathcal{U}1\ell}^k,D_{\mathcal{U}2\ell}^k\succeq O, D_{\mathcal{U}1\ell}^k,D_{\mathcal{U}2\ell}^k\text{ diagonal},\,\forall\,\ell\in\mathbb{N}_{[1,n_\mathcal{X}]},\label{eq:opt_e}\\
				& \footnotesize\begin{bmatrix}
					X_{I2j}^k & O & \left(A_{CL}^k\right)^\top\\
					* & \hspace{-1mm}H_{\mathcal{Z}}^\top D_{I2j}^kH_{\mathcal{Z}} & \left(B_{CL}^k\right)^\top\\
					* & * & X_{I1j}^k
				\end{bmatrix}\hspace{-1mm}\succ O,\footnotesize\begin{bmatrix}
					X_{O2j}^k & O & \left(A_{CL}^k\right)^\top\\
					* & \hspace{-1mm}H_{\mathcal{Z}}^\top D_{O2j}^kH_{\mathcal{Z}} & \left(B_{CL}^k\right)^\top\\
					* & * & X_{O1j}^k
				\end{bmatrix}\hspace{-1mm}\succ O,\forall j\in\mathbb{N}_{[1,n+n_K]},\hspace{-1mm}\label{eq:opt_f}\\
				& \footnotesize\begin{bmatrix}
					D_{I1j}^k & I_{n+n_K} \\ * & \left(\widetilde{H}_{I}^kY_{I2j}^k\right)^\top+\widetilde{H}_{I}^kY_{I2j}^k- \left(Y_{I2j}^k\right)^\top X_{I2j}^kY_{I2j}^k
				\end{bmatrix}\succ O,\,\forall\,j\in\mathbb{N}_{[1,n+n_K]},\label{eq:opt_g}\\
				& \footnotesize\begin{bmatrix}
					D_{O1j}^k & I_{n+n_K} \\ * & \left({H}_{O}^kY_{O2j}^k\right)^\top+{H}_{O}^kY_{O2j}^k- \left(Y_{O2j}^k\right)^\top X_{O2j}^kY_{O2j}^k
				\end{bmatrix}\succ O,\,\forall\,j\in\mathbb{N}_{[1,n+n_K]},\label{eq:opt_h}\\
				& \footnotesize\begin{bmatrix}
					\left(Y_{I1j}^k\right)^\top\widetilde{H}_{Ii}^k+\left(\widetilde{H}_{Ii}^k\right)^\top Y_{I1j}^k- \left(Y_{I1j}^k\right)^\top X_{I1j}^kY_{I1j}^k & e_j\\
					* & r_{Ij}
				\end{bmatrix}\succ O,\,\forall\,j\in\mathbb{N}_{[1,n+n_K]},\label{eq:opt_i}\\
				& r_{Ij}=2\eta-1_{n+n_K}^\top D_{I1j}^k1_{n+n_K}^{}-1_{n_{\mathcal{Z}}}^\top D_{I2j}^k1_{n_{\mathcal{Z}}}^{},\,\forall\,j\in\mathbb{N}_{[1,n+n_K]},\label{eq:opt_j}\\
				& \footnotesize\begin{bmatrix}
					\left(Y_{O1j}^k\right)^\top{H}_{Oi}^k+\left({H}_{Oi}^k\right)^\top Y_{O1j}^k- \left(Y_{O1j}^k\right)^\top X_{O1j}^kY_{O1j}^k & e_j\\
					* & r_{Oj}
				\end{bmatrix}\succ O,\,\forall\,j\in\mathbb{N}_{[1,n+n_K]},\label{eq:opt_k}\\
				& r_{Oj}=2-1_{n+n_K}^\top D_{O1j}^k1_{n+n_K}^{}-1_{n_{\mathcal{Z}}}^\top D_{O2j}^k1_{n_{\mathcal{Z}}}^{},\,\forall\,j\in\mathbb{N}_{[1,n+n_K]},\label{eq:opt_l}\\
				& \footnotesize\begin{bmatrix}
					D_{\mathcal{S}q}^k & \left(\begin{bmatrix}
						H_\mathcal{S} & O
					\end{bmatrix}\widetilde{H}_{Ii}^k\right)^\top e_q\\
					* & 2(\varepsilon_s-\varepsilon_m-\varepsilon_p)-1_{n+n_K}^\top D_{\mathcal{S}q}^k1_{n+n_K}
				\end{bmatrix}\succ O,\,\forall\,q\in\mathbb{N}_{[1,n_{\mathcal{S}}]},\label{eq:opt_m}\\
				& \footnotesize\begin{bmatrix}
					D_{\mathcal{X}p}^k & \left(\begin{bmatrix}
						H_\mathcal{X} & O
					\end{bmatrix}{H}_{Oi}^k\right)^\top e_p\\
					* & 2-1_{n+n_K}^\top D_{\mathcal{X}p}^k1_{n+n_K}
				\end{bmatrix}\succ O,\,\forall\,p\in\mathbb{N}_{[1,n_{\mathcal{X}}]},\label{eq:opt_n}\\
				& \footnotesize\begin{bmatrix}
					X_{\mathcal{U}\ell}^k & O & \left(H_{\mathcal{U}}C_{CL}^k\right)^\top e_j\\\vspace{-3mm}\\
					* & H_{\mathcal{Z}}^\top D_{\mathcal{U}2j}^kH_{\mathcal{Z}} & \left(H_{\mathcal{U}}D_{CL}^k\right)^\top e_j\\
					* & * & 2-1_{n+n_K}^\top D_{\mathcal{U}1j}^k1_{n+n_K}^{}-1_{n_{\mathcal{Z}}}^\top D_{\mathcal{U}2j}^k1_{n_{\mathcal{Z}}}^{}
				\end{bmatrix}\succ O,\,\forall\,\ell\in\mathbb{N}_{[1,n_{\mathcal{U}}]},\label{eq:opt_o}\\
				& \footnotesize\begin{bmatrix}
					D_{\mathcal{U}1j}^k & I_{n+n_K} \\ * & \left({H}_{O}^kY_{\mathcal{U}\ell}^k\right)^\top+{H}_{O}^kY_{\mathcal{U}\ell}^k- \left(Y_{\mathcal{U}\ell}^k\right)^\top X_{\mathcal{U}\ell}^kY_{\mathcal{U}\ell}^k
				\end{bmatrix}\succ O,\,\forall\,\ell\in\mathbb{N}_{[1,n_{\mathcal{U}}]},\label{eq:opt_p}\\
				& -1_{n+n_K}\leq H_O^k \begin{bmatrix}
					\omega_i^\top & 0_{n_K}
				\end{bmatrix}^\top \leq1_{n+n_K},\,\forall\,i\in\mathbb{N}_{[1,n_c]},\label{eq:opt_q}\\
				&A_{CL}^k=\footnotesize\begin{bmatrix}
					A+BD_K^k & BC_K^k \\ B_K^k & A_K^k
				\end{bmatrix},\,
				B_{CL}^k=\footnotesize\begin{bmatrix}
					I_n & BD_K^k \\ O & B_K^k
				\end{bmatrix},\,
				C_{CL}^k=\footnotesize\begin{bmatrix}
					D_K^k & C_K^k
				\end{bmatrix},\, D_{CL}^k=\footnotesize\begin{bmatrix}
					O & D_K^k
				\end{bmatrix},\label{eq:opt_r}\\
				& Q^k=O.\label{eq:opt_s}
			\end{empheq}
		\end{subequations}
		\hrulefill\vspace{-6mm}
	\end{figure*}   
	
	\begin{theorem}\label{thm:syn}
		Assume that the optimization problem in the initialization step of Algorithm~\ref{alg:syn_proc} (located on the next page) is feasible, for the provided input data. Then, we have that:
		\begin{enumerate}
			\item[a)] The optimization problem solved in the iterative phase of Algorithm~\ref{alg:syn_proc} is recursively feasible, and the procedure is guaranteed to converge;
			
			\item[b)] For any optimizer of $\mathcal{P}^k\left(Q^k\right)$ and any $k\in\mathbb{N}$, the matrices $\widetilde{H}_I^{k}$, $\widetilde{H}_{Ii}^{k}$, ${H}_O^{k}$ and ${H}_{Oi}^{k}$ are invertible;
			
			\item[c)] If the value of the cost function converges to $0$ (to within a given numerical tolerance) for some iteration $k^\star$, then by using $\widetilde{H}_I^{k^\star}$ to form \eqref{eq:OI_lifted} and ${H}_O^{k^\star}$ to form
			\begin{equation*}
				{\Omega}_O:=\left\{\xi\in\mathbb{R}^{n+n_K}:-1_{n+n_K} \leq {H}_O^{k^\star}\xi\leq
				1_{n+n_K} \right\},
			\end{equation*}
			then the closed-loop system in \eqref{eq:cl_dyn}-\eqref{eq:cl_real} obtained via $(A_{K}^{k^\star},B_{K}^{k^\star},C_{K}^{k^\star},D_{K}^{k^\star})$ satisfies C1)-C2) and O1)-O4).
			
		\end{enumerate}
	\end{theorem}
	
	\begin{proof}
		See the appendix.
	\end{proof}
	
	Before moving on to a numerical example, which showcases the applicability of our proposed procedure, we touch upon three crucial aspects related to Algorithm~\ref{alg:syn_proc}, namely:\vspace{-1mm}
	\begin{enumerate}
		\item \emph{Convexity}: As discussed in the proof of Theorem~\ref{thm:syn}, the entire purpose of our procedure's iterative phase is to convexify the bilinear constrains $\widetilde{H}_I\widetilde{H}_{Ii}=I_{n+n_K}$ and ${H}_O{H}_{Oi}=I_{n+n_K}$. By exploiting the direct multiplication between a pair of variables and by locking one of them in place via the constraint $Q^k=O$, the process of convexifying the optimization problem is done in a similar, yet less computationally expensive manner than the more general case tackled in \cite{aug_sparse}.
		
		\smallskip\vspace{-1mm}
		
		\item \emph{Convergence}: Although Algorithm~\ref{alg:syn_proc} is guaranteed to converge, note that the value of the cost function may not always converge to 0, even when a solution exists, as discussed in \cite{BMI2LMI} (from which our approach takes direct inspiration). As is often the case with non-convex optimization, the procedure's initialization is crucial, and this observation is of paramount importance for our proposed solution. To this end, we indicate the line search-based initialization technique presented at the end of Section IV in \cite{svi2} as being a reliable remedy.
		
		\smallskip
		
		\item \emph{Connections}: Bilinear equality constraints aside, the convexification procedure employed in the proof of Theorem~\ref{thm:syn} is most closely associated with the \emph{slack variable identity} used in \cite{svi1} and in \cite{svi2} to linearize a series of matrix inequalities. However, we point out that the direct convexification performed in these papers (without requiring an iterative phase, as in Algorithm~\ref{alg:syn_proc}) is owed to the absence of any state measurement noise ($v[k]$, in our case) and to the joint design of the control law in tandem with \emph{a single invariant set}. When two sets appear in the formulation, the need to ensure that \emph{the same} state feedback ensures invariance for both of them forces the inclusion of bilinear equality constraints. Similar complications arise when considering non-zero measurement noise.
	\end{enumerate}
	
	This being said, we now present an application for our proposed control algorithm.
	
	\section{Numerical Example}\label{sec:num_ex}
	
	We present here a simple yet illustrative numerical example that highlights the practical potential of the proposed design framework. For a sampling time of $T_s=0.1$ seconds, consider the discretized equations
	\begin{equation}\label{eq:sys_ex}
		\left\{\begin{aligned}
			&\begin{bmatrix}
				\Delta p[k+1] \\ \Delta s[k+1]
			\end{bmatrix} = \begin{bmatrix}
				1 & T_s \\ 0 & 1
			\end{bmatrix} \begin{bmatrix}
				\Delta p[k] \\ \Delta s[k]
			\end{bmatrix} + \sigma[k]\begin{bmatrix}
				\frac{T_s^2}{2} \\ T_s
			\end{bmatrix}u[k] +\\ &\qquad+(1-\sigma[k])\begin{bmatrix}
				\frac{T_s^2}{2} \\ T_s
			\end{bmatrix}u_d[k] + \begin{bmatrix}
				\frac{T_s^2}{2} \\ T_s
			\end{bmatrix} (w_a[k]-a_r[k]),\\
			&\begin{bmatrix}
				y_p[k]  \\ y_v[k]
			\end{bmatrix} = \begin{bmatrix}
				\Delta p[k] \\ \Delta s[k]
			\end{bmatrix} + \begin{bmatrix}
				v_p[k] \\ v_s[k]
			\end{bmatrix},
		\end{aligned}\right.
	\end{equation}
	which are dynamics of type \eqref{eq:ss_sys}, and where we have that:\smallskip
	\begin{enumerate}
		\item[a)] $\Delta p[k]$ and $\Delta s[k]$ are the relative position and speed of an idealized vehicle, with respect to a virtual reference body (whose position and speed are initialized in $0$);\medskip
		
		\item[b)] $u[k]$ represents the acceleration implemented by the cruise-control system implemented via \eqref{eq:ss_ctl};\smallskip
		
		\item[c)] $u_d[k]$ is the acceleration being generated by the decisions of the car's human driver;\smallskip
		
		\item[d)] $a_r[k]$ is the acceleration of the virtual reference body, which is computed by some reference governor;\smallskip
		
		\item[e)] $w_a[k]$ models acceleration disturbance produced by uneven road conditions;\smallskip
		
		\item[f)] $y_p[k]$ and $y_s[k]$ are the state measurements affected by the sensor noise, modelled via $v_p[k]$ and $v_s[k]$.\smallskip
	\end{enumerate}
	
	To this dynamical model, we associate the sets discussed in Section~\ref{subsec:prob_st}. We denote $B:=\begin{bmatrix}
		\frac{T_s^2}{2} & T_s
	\end{bmatrix}^\top$ and we assume that the actions of the driver are limited to $50\%$ of the car's maximum capacity, to define $\mathcal{U}:=\{u\in\mathbb{R}:|u|\leq 50\}$ and $\mathcal{D}:=\{Bu_d\in\mathbb{R}^2:|u_d|\leq 25\}$. We treat $a_r[k]$ as a disturbance, whose impact on $\Delta p[k]$ and $\Delta s[k]$ must be managed by either the driver or by the cruise-control system, in order to define the set
	\begin{equation*}
		\mathcal{W}:=\left\{\begin{bmatrix}
			-B \ B
		\end{bmatrix}\begin{bmatrix}
			a_r \ w_a
		\end{bmatrix}^\top\hspace{-1mm}\in\mathbb{R}^2:|a_r|\leq 1,\ |w_a|\leq0.5\right\}.
	\end{equation*} 
	For the measurement noise, we consider
	\begin{equation*}
		\mathcal{V}:=\left\{\begin{bmatrix}
			v_p\ v_s
		\end{bmatrix}^\top\hspace{-1mm}\in\mathbb{R}^2:|v_p|\leq0.01,\ |v_s|\leq0.01\right\},
	\end{equation*}
	which enables us to choose $\varepsilon_p=\varepsilon_m=0.01$, along with
	\begin{equation*}
		\mathcal{S}:=\left\{\begin{bmatrix}
			\Delta_p\ \Delta_s
		\end{bmatrix}^\top\hspace{-1mm}\in\mathbb{R}^2:|\Delta_p|\leq1,\ |\Delta_s|\leq1\right\}.
	\end{equation*}
	Finally, we select the set
	\begin{equation}\label{eq:X_set}
		\mathcal{X}:=\left\{\begin{bmatrix}
			\Delta_p\ \Delta_s
		\end{bmatrix}^\top\hspace{-1mm}\in\mathbb{R}^2:|\Delta_p|\leq10^3,\ |\Delta_s|\leq10^2\right\}.
	\end{equation}
	It is straightforward to check that \eqref{eq:incl_a}-\eqref{eq:incl_c} hold. Thus, we employ the procedure discussed in Section~\ref{subsec:syn_proc} to obtain both the controller from \eqref{eq:ss_ctl} and the pair of sets $\Omega_I$ and $\Omega_O$.

	\begin{remark}
		We point out that the large values chosen in \eqref{eq:X_set} for the components of $\mathcal{X}$ are selected in this manner only for the sake of convenience, in order to allow for a straightforward feasible initialization of the iterative synthesis procedure presented in Section~\ref{subsec:syn_proc}. As shown in the sequel, this particular choice is in no way conservative, since the obtained $\Omega_O$ set is located far away from the frontier of $\mathcal{X}$ and it forms a tight outer approximation of $\mathcal{S}_c\,$.
	\end{remark}
	By selecting $n_K=2$, $\varepsilon_s=0.6$ and $\eta = 0.99$, the aforementioned procedure produces a dynamical controller having the following realization:
	\begin{equation}\label{eq:real_ex}
		\hspace{-2mm}\begin{array}{ll}
			A_K = \left[\footnotesize\begin{array}{rr}
				-0.1155  & -0.5041 \\ 0.0831  &  0.2651
			\end{array}\right],&\hspace{-2mm} C_K =\hspace{0.5mm} \left[\footnotesize\begin{array}{rr}
				\phantom{-}0.1478  \\  0.5393
			\end{array}\right]^\top,\\
			B_K = \left[\footnotesize\begin{array}{rr}
				-0.0183  & -0.1109 \\ -0.2731  & -0.0135
			\end{array}\right],&\hspace{-2mm} D_K = \left[\footnotesize\begin{array}{rr}
				-6.5049  \\ -8.7258
			\end{array}\right]^\top, 
		\end{array}\hspace{-2mm}
	\end{equation}
	along with the sets $\Omega_I = \left\{x\in\mathbb{R}^2:-1_6\leq H_I\, x \leq 1_6\right\}$, $\Omega_O = \left\{\xi\in\mathbb{R}^4:-1_4\leq H_O\, \xi \leq 1_4\right\}$ and the one in \eqref{eq:OI_lifted}, where we have:
	\begin{equation*}
		\begin{aligned}
			H_I=&\left[\footnotesize\begin{array}{rr}
				\phantom{-}0.8423  &  2.1526\\
				2.7910  &  0.1968\\
				2.6817  &  0.1069\\
				1.1249  & -0.9540\\
				2.5346  &  0.0000\\
				0.0000   & 1.6716
			\end{array}\right],\\
			\widetilde{H}_I=&\left[\footnotesize\begin{array}{rrrr}
				3.0373  &  \phantom{-}0.2146  &  \phantom{-}0.0074  &  \phantom{-}0.0249\\
				0.9882  &  2.5404  &  0.0724  &  0.0430\\
				-0.0131  &  0.0180  &  0.8019  &  0.1617\\
				-0.1606  &  0.0056  &  0.2297  &  1.1469
			\end{array}\right],\\
			H_O=&\left[\footnotesize\begin{array}{rrrr}
				0.2318  &  \phantom{-}0.0169  & -0.0003  & -0.0027\\
				0.1243  &  0.1934  & -0.0033  & -0.0118\\
				0.0038  &  0.0066  &  0.6959  &  0.1386\\
				-0.0021  &  0.0062  &  0.0626  &  0.5940
			\end{array}\right].
		\end{aligned}
	\end{equation*}
	
	Once again, it is straightforward to check that, whenever $\sigma[k]=1$, the closed-loop system formed by the dynamics from \eqref{eq:sys_ex} and the controller whose realization is given in \eqref{eq:real_ex} satisfies the conditions stated in Section~\ref{sec:main_res} for the given sets. In order to test the obtained control laws, we perform a simulation in which the signals $w_a[k]$, $v_p[k]$ and $v_s[k]$ are generated randomly, with uniform distribution, in their associated sets, while $a_r[k]$ and $u_d[k]$, which are the chief factors that contribute to the simulation scenario, are depicted in Figure~\ref{fig:accels}. By applying Algorithm~\ref{alg:ho_periodic}, the state of the system given in \eqref{eq:sys_ex} evolves as depicted in Figure~\ref{fig:state_evol}, in which the $\Delta$ markers indicate the activation of the feedback controller, while the $\nabla$ markers indicate its deactivation. The type of functioning described in Section~\ref{subsec:prob_st} is ensured by our control laws, with the resulting command and controller switching\footnote{For the signal $\sigma[k]-\sigma[k-1]$ in Figure~\ref{fig:cmd_and_switch}, a value of $1$ marks controller activation, $-1$ indicates controller deactivation, and $0$ represents no change.} signals being shown in Figure~\ref{fig:cmd_and_switch}, located at the top of the next page. {Note that $\sigma[0]=0$ and, therefore, the plant is in open-loop configuration from initialization up until $k=60$, as illustrated in the lower half of Figure~\ref{fig:cmd_and_switch}. Moreover, the runtime information associated with Algorithm~\ref{alg:ho_periodic} is presented in Table~\ref{tab:run_tim}, based on the 300 samples that form the simulation scenario depicted in Figures~\ref{fig:accels},~\ref{fig:state_evol}, and~\ref{fig:cmd_and_switch}. The values from this table show that all computation times are well below 1 millisecond in duration.}\vspace{-2mm}
	
	{
		\begin{table}[H]
			\centering
			\begin{tabular}{c|ccccc}
				Computation time&min&max&mean&median&mode\\\hline\hline
				Duration (milliseconds)&0.001&0.053&0.006&0.005&0.005
			\end{tabular}
			\caption{} 
			\label{tab:run_tim}
		\end{table}
	}\phantom{}\vspace{-17mm}
	
	\begin{figure}[H]
		\centering
		\includegraphics[width=\columnwidth]{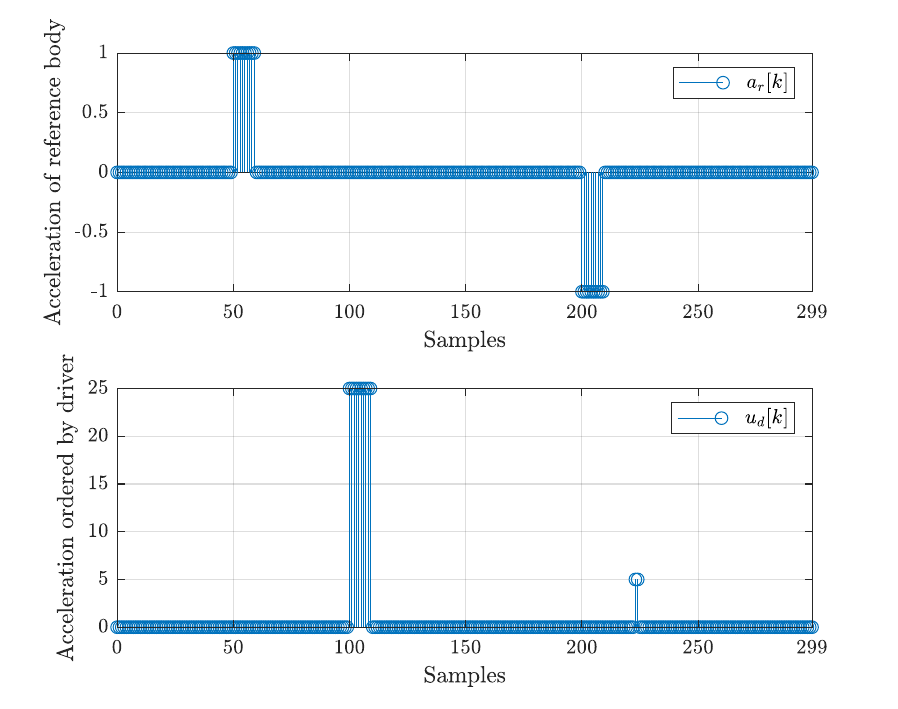}
		\caption{The external acceleration profiles that act upon the vehicle.}
		\label{fig:accels}
	\end{figure}\phantom{ }\vspace{-7mm}
	
	\begin{remark}
		The driver-induced burst of positive acceleration between $k=224$ and $k=225$ was timed to coincide with the controller deactivation at $k=224$. Notice, however, that the controller does not switch back on. This is owed to the scaling parameter $\varepsilon_s$ which, as previously explained in Remark~\ref{rem:eps}, has been employed in order to place distance between the frontiers of $\Omega_I$ and $\mathcal{S}$. As shown in the presented simulation scenario, this distance can indeed help prevent rapid oscillations between the system's two configurations.\vspace{-5mm}
	\end{remark}
	
	\begin{figure}[H]
		\centering
		\includegraphics[width=\columnwidth]{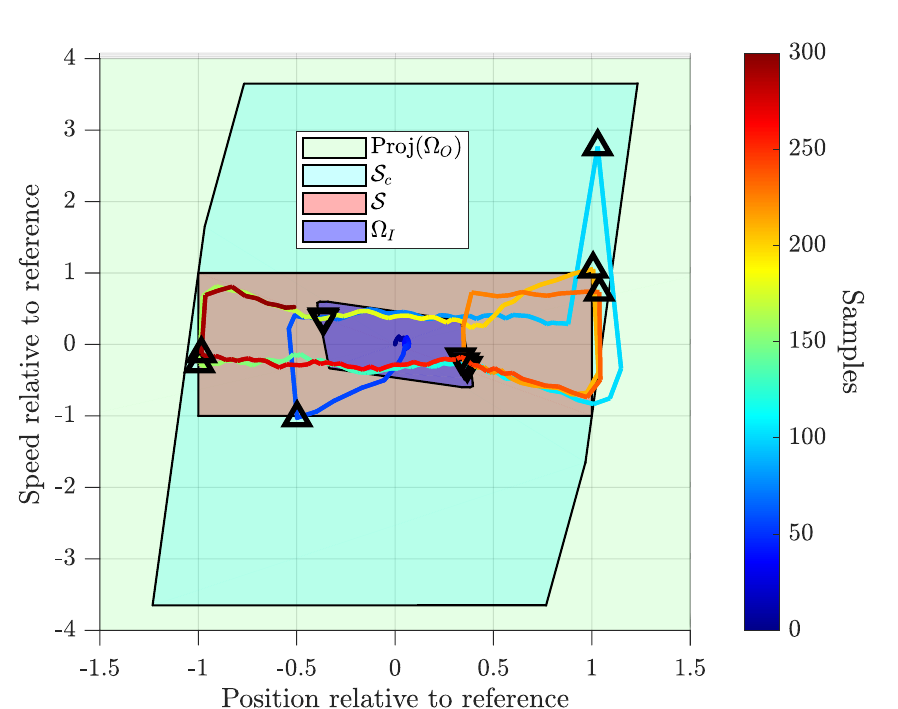}
		\caption{State evolution of the vehicle, under the action of Algorithm~\ref{alg:ho_periodic}.}
		\label{fig:state_evol}
	\end{figure}
	
	\begin{figure}[H]
		\centering
		\includegraphics[width=\columnwidth]{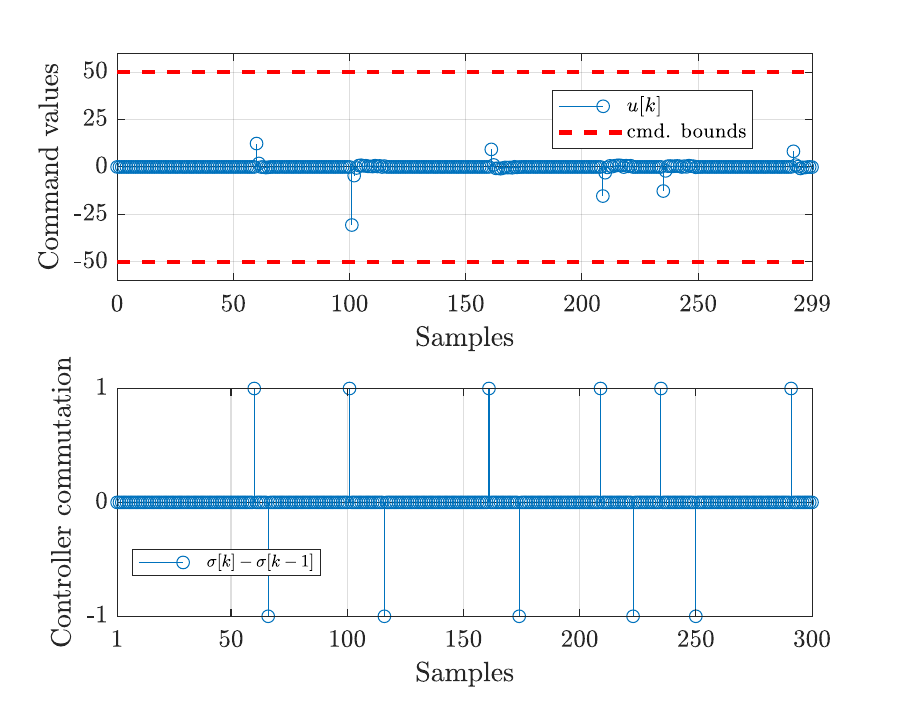}
		\caption{The feedback actuation and switching computed by Algorithm~\ref{alg:ho_periodic}.}
		\label{fig:cmd_and_switch}
	\end{figure}    
	
	\section{Conclusions and Future Work}\label{sec:outro}
	
	The concept of hands-off control was shown to extend naturally from the time-based perspective presented in \cite{ho_book} and \cite{ho_art} towards a set-based spatial sense, in which an inner-outer set pair hedges the frontier of the control (de)activation set. Crucially, the use of dynamical state-space-based controllers opens up a particularly tantalizing avenue for future research. By leveraging the distributed control law parametrization formalized in \cite{NRF}, the type of hands-off strategy proposed in this paper could be expanded beyond its current centralized setting. Indeed, the possibility of having multiple distributed sub-controllers that trigger independently and act only on local variables, as in Algorithm~\ref{alg:ho_periodic} (without the need for consensus-based mechanisms \cite{ho_dist}), can be a powerful asset whenever centralized decision policies are impractical to implement.

	\section*{APPENDIX}
	
	\textbf{Proof of Theorem~\ref{thm:algo}}\smallskip
	
	We begin by proving point c). Notice that, in Algorithm~\ref{alg:ho_periodic}, $\sigma[k]$ is set to $0$ only if $y[k]\in(\mathcal{S}\ominus(-\mathcal{V}))$ (which is non-empty, recalling Remark~\ref{rem:measurement}). Since $x[k] = y[k] + (-v[k])$, it follows that the aforementioned inclusion implies the fact that
	\begin{equation*}
		x[k]\in((\mathcal{S}\ominus(-\mathcal{V}))\oplus(-\mathcal{V}))\subseteq\mathcal{S}.
	\end{equation*}
	
	To prove point b), we will treat the twin cases concerning the latter in a unitary fashion. Indeed, by assumption, we have that $x[k_0]\in\mathcal{S}_c$\,. On the other hand, we have that $\sigma[k-1]=0$ for the time instants $k>k_0$, from which point c) ensures that $x[k-1]\in\mathcal{S}$ and, additionally, that $x[k]\in\mathcal{S}^+\subseteq\mathcal{S}_c$\,, due to \eqref{eq:ss_sys} and to $\sigma[k-1]=0$. Therefore, let $k_{\bullet}$ denote either $k_0$ or $k$, allowing us to state that $x[k_{\bullet}]\in\mathcal{S}_c$. In either of these two cases, Algorithm~\ref{alg:ho_periodic} mandates that the controller's state be initialized as $x_K[k_{\bullet}]=0_{n_K}$ and that the \textbf{Control} step be iterated until there exists some $T_{\bullet}\in\mathbb{N}_{>0}$, for which $y[k_{\bullet}+T_{\bullet}]\in(\Omega_I\ominus(-\mathcal{V}))$. Note also that the latter set is non-empty, due to I2) and to Assumption~\ref{asu:sets}. It follows that, until the latter inclusion is satisfied, the system from \eqref{eq:ss_sys} evolves 
	
	\newpage\noindent
	in closed-loop configuration with the one from \eqref{eq:ss_ctl}, as per the dynamics given in \eqref{eq:cl_dyn}. In this configuration, the system's state can be expressed, for some $j\in\mathbb{N}_{>0}$\,, via the identities
	\begin{equation}\label{eq:cl_func}
		\left\{
		\begin{aligned}
			x[k_\bullet+j]\phantom{:}=&\ x_i[k_\bullet+j] + x_z[k_\bullet+j]\,,\\
			x_i[k_\bullet+j]:=&\ \begin{bmatrix}
				I_n \ O
			\end{bmatrix}A_{CL}^j \begin{bmatrix}
				x^\top[k_{\bullet}] & 0_{n_K}^\top
			\end{bmatrix}^\top\hspace{-1mm},\\
			x_z[k_\bullet+j]:=&\ \sum_{i=1}^j\begin{bmatrix}
				I_n \ O
			\end{bmatrix}A_{CL}^{i-1}B_{CL}^{}z[k_{\bullet}+i-1]\,.
		\end{aligned}
		\right.\hspace{-3mm}
	\end{equation}
	
	In order to show that $\sigma[k]$ switches from $1$ to $0$ in at most $T_{\mathrm{max}}$ time instants, for some $T_{\mathrm{max}}\in\mathbb{N}_{>0}$, we first denote
	\begin{equation}\label{eq:ZN_def}
		\mathcal{Z}_N:=\bigoplus_{i=0}^N\begin{bmatrix}
			I_n \ \, O
		\end{bmatrix}A_{CL}^iB_{CL}^{ }\mathcal{Z},\,\forall\, N\in\mathbb{N}\,,
	\end{equation}
	and we point out that, since $0_{2n}\in\mathcal{Z}:=\mathcal{W}\times\mathcal{V}$ (due to Assumption~\ref{asu:sets}), it follows that $\mathcal{Z}_{N_1}\subseteq\mathcal{Z}_{N_2}$, for any two integers $0\leq N_1\leq N_2$. Using this fact in conjunction with I3), we obtain the following inclusions
	\begin{equation}\label{eq:forced_incl}
		x_z[k_{\bullet}+j]\in\mathcal{Z}_j\subseteq\beta(\Omega_I\ominus\mathcal{N}),\,\forall\,j\in\mathbb{N}_{>0}\,.
	\end{equation}

	Define $\mu:=\sup_{r\in\mathcal{S}_c}\|r\|$, which is finite, due to Assumption~\ref{asu:sets}, and strictly positive, due to \eqref{eq:incl_b}, in order to notice that $\|x_i[k_{\bullet}+j]\|\leq\mu\|A_{CL}^j\|$. Recalling that $A_{CL}$ is a Schur matrix, it follows that $\lim_{j\rightarrow\infty}\|A_{CL}^j\|=0$. Consequently, there exists $T_{\mathrm{max}}$ such that $\|A_{CL}^j\|\leq\frac{\alpha(1-\beta)}{\mu},\forall\,j\geq T_{\mathrm{max}}$, which further implies $x_i[k_{\bullet}+j]\in(1-\beta)\,\overline{\mathcal{B}}_{\alpha},\forall\,j\geq T_{\mathrm{max}}$. Combining the latter inclusion with I2), we get that
	\begin{equation}\label{eq:free_incl}
		x_i[k_{\bullet}+T_{\mathrm{max}}]\in(1-\beta)(\Omega_I\ominus\mathcal{N}).
	\end{equation}
	
	We now proceed to embed \eqref{eq:forced_incl} and \eqref{eq:free_incl} into \eqref{eq:cl_func}, and we employ the polytopic properties stated in Assumption~\ref{asu:sets} to obtain that $x[k_{\bullet}+T_{\mathrm{max}}]\in(\Omega_I\ominus\mathcal{N})$. Since the identity $y[k_{\bullet}+T_{\mathrm{max}}]=x[k_{\bullet}+T_{\mathrm{max}}]+v[k_{\bullet}+T_{\mathrm{max}}]$ holds, we get $y[k_{\bullet}+T_{\mathrm{max}}]\in(\Omega_I\ominus(-\mathcal{V}))\subseteq(\mathcal{S}\ominus(-\mathcal{V}))$, where the latter set inclusion follows from I1), $\varepsilon_s\in(0,1)$ and Assumption~\ref{asu:sets}. By recalling now the conditional statements located in Algorithm~\ref{alg:ho_periodic} just after the \textbf{Monitoring} step, we finally conclude that $\sigma[k_\bullet+T_{\mathrm{max}}]$ will always be set to $0$, provided the system's state has not been brought to $(\Omega_I\ominus\mathcal{N})$ at an earlier time instant $k_{\bullet}+T_{\bullet}$\,, for some $T_{\bullet}\in\mathbb{N}_{[1,T_\mathrm{max}]}$.
	
	We conclude the proof by showing that the statement from point a) holds. Indeed, whenever $\sigma[k]=0$, we have $x[k]\in\mathcal{X}$, due to point c) and to $\mathcal{S}\subset\mathcal{X}$. Moreover, direct inspection of Algorithm~\ref{alg:ho_periodic} shows that whenever $\sigma[k]$ is set to $0$, we have that $u[k]=0_m\in\mathcal{U}$, due to Assumption~\ref{asu:sets}. 
	
	In order to show that the same inclusions hold for those time instants when $\sigma[k]=1$, we employ the properties from points O1)-O4). We reuse the $k_{\bullet}$ employed in the proof of point b) to state that $\begin{bmatrix}
		x^\top[k_{\bullet}] &0^\top_{n_K}
	\end{bmatrix}\hspace{-1mm}{}^\top\in\mathcal{S}_c\times\{0_{n_K}\}\subseteq\Omega_O$, with the latter set inclusion being owed to O1). Moreover, as discussed in the proof of point b), the closed-loop system evolves according to the dynamics from \eqref{eq:cl_dyn} for as long as $\sigma[k]=1$. By employing now O2), we can state the fact that $\sigma[k]=1\Rightarrow\xi[k]\in\Omega_O$ and, since $x[k]=\begin{bmatrix}
		I_n & O
	\end{bmatrix}\xi[k]$, the set inclusion from O3) further implies that if $\sigma[k]=1$, then $x[k]\in\mathcal{X}$. Similarly, the implication $\sigma[k]=1\Rightarrow u[k]\in\mathcal{U}$ follows from \eqref{eq:cl_dyn} and from the set inclusion given in O4). \qed\smallskip
	
	\textbf{Proof of Theorem~\ref{thm:set_design}}\smallskip
	
	To prove point a), we denote $\widetilde{\mathcal{Z}}_0:=B_{CL}\,\mathcal{Z}$ in order to recursively define $\widetilde{\mathcal{Z}}_{N+1}:=A_{CL}\widetilde{\mathcal{Z}}_N\oplus\widetilde{\mathcal{Z}}_0$, for all $N\in\mathbb{N}$. Due to \eqref{eq:OI_lifted} and to C1), we have $\widetilde{\mathcal{Z}}_0\subseteq\eta\,\widetilde{\Omega}_I\subseteq\widetilde{\Omega}_I$. Then, it follows that $\widetilde{\mathcal{Z}}_N\subseteq\eta\,\widetilde{\Omega}_I$, for all $N\in\mathbb{N}$. Recalling the sets defined in \eqref{eq:ZN_def} and employing C2) along with the fact that $\begin{bmatrix}
		I_n & O
	\end{bmatrix}(\eta\,\widetilde{\Omega}_I) = \eta\big(\begin{bmatrix}
		I_n & O
	\end{bmatrix}\widetilde{\Omega}_I\big)$ (this follows by describing the polytope $\widetilde{\Omega}_I$ in its vertex-based representation), we obtain the following inclusions
	\begin{equation}\label{eq:proj_incl}
		\frac{1}{\eta}\mathcal{Z}_N\subseteq\begin{bmatrix}
			I_n & O
		\end{bmatrix}\widetilde{\Omega}_i\subseteq(\varepsilon_s-\varepsilon_p-\varepsilon_m)\,\mathcal{S},\,\forall\,N\in\mathbb{N}.
	\end{equation}
	By taking the inclusion on the left-hand side of \eqref{eq:proj_incl} and by recalling the fact that $\Omega_I = \begin{bmatrix}
		I_n & O
	\end{bmatrix}\widetilde{\Omega}_I\oplus\mathcal{N}$, we employ the pair of equivalent conditions
	\begin{equation}\label{eq:eq_cond_I3}
		\frac{1}{\eta}\mathcal{Z}_N\oplus\mathcal{N}\subseteq\Omega_I\iff\frac{1}{\eta}\mathcal{Z}_N\subseteq\Omega_I\ominus\mathcal{N},
	\end{equation}
	and we scale the inclusion from the right-hand side of \eqref{eq:eq_cond_I3} by $\eta$ to get that I3) holds for $\beta=\eta$. Similarly, by taking now the right-hand inclusion given in \eqref{eq:proj_incl} and by employing \eqref{eq:incl_a}, it is straightforward to obtain $\Omega_I = \begin{bmatrix}
		I_n & O
	\end{bmatrix}\widetilde{\Omega}_I\oplus\mathcal{N}\subseteq\varepsilon_s\,\mathcal{S}$, from which we retrieve I1).
	
	Moving on, we will construct a closed ball inside $\widetilde{\Omega}_I$ and, by using projection-based arguments, we will retrieve I2). Let $\psi:=\frac{1}{\left\|\widetilde{H}_I\right\|}$, along with $\overline{\mathcal{B}}_{\psi}:=\{\xi\in\mathbb{R}^{n+n_K}:\|\xi\|\leq\psi\}$. Since we have that $\left\|\widetilde{H}_I\xi\right\|\leq\left\|\widetilde{H}_I\right\|\left\|\xi\right\|\leq1$ for all $\xi\in\overline{\mathcal{B}}_\psi$, it is straightforward to notice that $\overline{\mathcal{B}}_\psi\subseteq\widetilde{\Omega}_I$ and, moreover, that $\begin{bmatrix}
		I_n & O
	\end{bmatrix}\overline{\mathcal{B}}_\psi\subseteq\begin{bmatrix}
		I_n & O
	\end{bmatrix}\widetilde{\Omega}_I$. Consider now the closed ball $\overline{\mathcal{B}}_{\psi}^{\,p}:=\{x\in\mathbb{R}^{n}:\|x\|\leq\psi\}$ in order to state that, for all $\xi\in\overline{\mathcal{B}}_{\psi}$, we have that $\|\begin{bmatrix}
		I_n & O
	\end{bmatrix}\xi\|\leq\|\xi\|\leq\psi$ from which the inclusion $\begin{bmatrix}
		I_n & O
	\end{bmatrix}\overline{\mathcal{B}}_\psi\subseteq\overline{\mathcal{B}}_{\psi}^{\,p}$ follows. In order to prove that the aforementioned inclusion is an equality, consider first an arbitrary $x\in\overline{\mathcal{B}}_{\psi}^{\,p}$. Then, there exists $\widetilde{\xi}:=\begin{bmatrix}
		x^\top & 0_{n_K}^\top
	\end{bmatrix}^\top$ such that $x=\begin{bmatrix}
		I_n & O
	\end{bmatrix}\widetilde{\xi}$ and that $\widetilde{\xi}\in\overline{\mathcal{B}}_{\psi}$, which proves that $\overline{\mathcal{B}}_{\psi}^{\,p}\subseteq\begin{bmatrix}
		I_n & O
	\end{bmatrix}\overline{\mathcal{B}}_\psi$. By employing the following chain of equalities and inclusions
	\begin{multline*}
		\overline{\mathcal{B}}_{\psi}^{\,p}=\begin{bmatrix}
			I_n & O
		\end{bmatrix}\overline{\mathcal{B}}_\psi\subseteq\begin{bmatrix}
			I_n & O
		\end{bmatrix}\widetilde{\Omega}_I\subseteq\\
		\subseteq\left(\begin{bmatrix}
			I_n & O
		\end{bmatrix}\widetilde{\Omega}_I\oplus\mathcal{N}\right)\ominus\mathcal{N}=\Omega_I\ominus\mathcal{N},
	\end{multline*}
	we finally retrieve I2) for $\alpha=\psi=\frac{1}{\left\|\widetilde{H}_I\right\|}$.\smallskip
	
	In order to prove point b), we begin by showing that $\lim\limits_{k\rightarrow\infty}\left\|A_{CL}^k\xi\right\|=0$ for all $\xi\in\mathbb{R}^{n+n_K}$ and then we prove (by contradiction) that $A_{CL}$ cannot have eigenvalues on or outside of the unit circle. To this end, note that for all $\xi\in\widetilde{\Omega}_I$, we have $\left\|\widetilde{H}_I\xi\right\|\leq\sqrt{n+n_K}$. Thus, recalling that $\det\big(\widetilde{H}_I \big)\neq0$ from \eqref{eq:OI_lifted} and defining $\theta:=\left\|\widetilde{H}_I^{-1}\right\|\sqrt{n+n_K}$, it follows that
	\begin{equation*}
		\|\xi\|=\left\|\widetilde{H}_I^{-1}\left(\widetilde{H}_I\xi\right)\right\|\leq\theta,\,\forall\,\xi\in\widetilde{\Omega}_I.
	\end{equation*}
	Using now the fact that $0_{2n}\in\mathcal{Z}$ (due to Assumption~\ref{asu:sets}) along with C1), to state that $A_{CL}\widetilde{\Omega}_I\subseteq\eta\,\widetilde{\Omega}_I$ which further implies that $A_{CL}^k\widetilde{\Omega}_I\subseteq\eta^k\widetilde{\Omega}_I$ for all $k\in\mathbb{N}$. We conclude, therefore, that $\left\|A_{CL}^k\xi\right\|\leq\eta^k\theta$ for all $\xi\in\widetilde{\Omega}_I$ and all $k\in\mathbb{N}$.
	
	In order to extend this property to any $\xi\in\mathbb{R}^{n+n_K}$ and retrieve the sought-after limit value, we recall the set $\overline{\mathcal{B}}_\psi$ introduced earlier and we point out that for any $\xi\in\mathbb{R}^{n+n_K}$, it is possible to define $\widehat{\xi}:=\dfrac{\psi}{\|\xi\|}\xi\in\overline{\mathcal{B}}_\psi\subseteq\widetilde{\Omega}_I$. By employing this new vector, it becomes possible to state that
	\begin{equation*}
		0\leq\left\|A_{CL}^k\xi\right\|=\frac{\|\xi\|}{\psi}\left\|A_{CL}^k\widehat{\xi}\right\|\leq\frac{\eta^k\theta\|\xi\|}{\psi}.
	\end{equation*}
	Recalling that $\eta\in(0,1)$ and employing the squeeze theorem, it follows that $\lim\limits_{k\rightarrow\infty}\left\|A_{CL}^k\xi\right\|=0$ for all $\xi\in\mathbb{R}^{n+n_K}$.
	
	In order to conclude the proof, let $A_{CL}=PJP^{-1}$ where $J$ is in Jordan canonical form and assume that there exists $i\in\mathbb{N}_{[1,n+n_K]}$ such that $Je_i=\lambda_i e_i$ with $|\lambda_i|\geq 1$, where $e_i$ denotes the $i^{\text{th}}$ column of $I_{n+n_K}$. Then, by defining the vector $p_i:=Pe_i$, we must have that
	\begin{equation*}
		\left\|A_{CL}^kp_i\right\|=\left\|PJ^ke_i\right\|=\left\|\lambda_i^kp_i\right\|\geq\|p_i\|,\,\forall\,k\in\mathbb{N}.
	\end{equation*}
	Note that $\|p_i\|>0$, due to $P$ being invertible, which means that there cannot exist any $k\in\mathbb{N}$ such that $\left\|A_{CL}^kp_i\right\|\leq\frac{\|p_i\|}{2}$. However, since this statement would contradict the fact that $\lim\limits_{k\rightarrow\infty}\left\|A_{CL}^kp_i\right\|=0$, then $A_{CL}$ cannot have any eigenvalue $\lambda_i$ with $|\lambda_i|\geq 1$ and is, therefore, a Schur matrix. \qed\smallskip
	
	\textbf{Proof of Theorem~\ref{thm:syn}}\smallskip
	
	We begin by showing the recursive feasibility of Algorithm~\ref{alg:syn_proc}. To this end, assume that $\mathcal{P}^0(Q^0)$ is feasible. Then, by employing the same notation for an optimizer of this problem as for the decision variables appearing in \eqref{eq:opt_a}-\eqref{eq:opt_s} for $k=0$, we are able to assert that
	\eqref{eq:opt_f} and \eqref{eq:opt_o} imply $X_{I1j}^k, X_{I2j}^k, X_{O1j}^k, X_{O2j}^k, X_{\mathcal{U}\ell}^k\succ O$, which also makes all of these matrices invertible. Therefore, it is possible to compute $Y_{I1j}^{k+1}$, $Y_{I2j}^{k+1}$, $Y_{O1j}^{k+1}$, $Y_{O2j}^{k+1}$ and $Y_{\mathcal
		U\ell}^{k+1}$ as in Algorithm~\ref{alg:syn_proc}. We now show (for $k=0$) that any optimizer of $\mathcal{P}^k\left(Q^k\right)$ is feasible for $\mathcal{P}^{k+1}\left(Q^{k+1}\right)$. Indeed, notice that the inequalities and the sparsity constraints given in \eqref{eq:opt_b}-\eqref{eq:opt_f} and \eqref{eq:opt_m}-\eqref{eq:opt_q} are trivially satisfied by an optimizer of $\mathcal{P}^k\left(Q^k\right)$, whereas \eqref{eq:opt_s} is always satisfied by an optimizer of the previous iteration (for either of its two branches). Thus, we need only check that \eqref{eq:opt_g}-\eqref{eq:opt_i} along with \eqref{eq:opt_k} and \eqref{eq:opt_p} are satisfied by any optimizer of $\mathcal{P}^k\left(Q^k\right)$, and we begin with~\eqref{eq:opt_g}.
	
	By recalling that $X_{I2j}^k\succ O$, for all $j\in\mathbb{N}_{[1,n+n_K]}$, and by plugging the aforementioned optimizer along with the explicit expression of $Y_{I2j}^{k+1}$ into \eqref{eq:opt_g}, we get the fact that\vspace{-1mm}
	\begin{equation}\label{eq:svi_a1}
		\begin{aligned}
			M_{I2j}^{k+1}:=&\scriptsize\begin{bmatrix}
				D_{I1j}^k & I_{n+n_K} \\ * &\hspace{-3mm} \left(\widetilde{H}_{I}^kY_{I2j}^{k+1}\right)^\top\hspace{-1.5mm}+\widetilde{H}_{I}^kY_{I2j}^{k+1}- \left(Y_{I2j}^{k+1}\right)^\top \hspace{-1.5mm}X_{I2j}^kY_{I2j}^{k+1}
			\end{bmatrix}\\
			=&\footnotesize\begin{bmatrix}
				D_{I1j}^k & I_{n+n_K} \\ * &\hspace{-2.5mm} \widetilde{H}_I^k\hspace{-0.5mm} \left(X_{I2j}^k\right)^{-1}\hspace{-1mm}\left(\widetilde{H}_I^k\right)^\top
			\end{bmatrix}\\
			=&\scriptsize\begin{bmatrix}
				D_{I1j}^k & I_{n+n_K} \\ * &\hspace{-3mm} \left(\widetilde{H}_{I}^kY_{I2j}^{k}\right)^\top\hspace{-1.5mm}+\widetilde{H}_{I}^kY_{I2j}^{k}- \left(Y_{I2j}^{k}\right)^\top \hspace{-1.5mm}X_{I2j}^kY_{I2j}^{k}+Z_{I2j}^k
			\end{bmatrix}\hspace{-0.5mm},\\\phantom{}
		\end{aligned}
	\end{equation}
	for all $j\in\mathbb{N}_{[1,n+n_K]}$, where the last equality follows from the \emph{slack variable identity} (see (41) in \cite{svi2}), having defined\vspace{-1mm}
	\begin{equation}\label{eq:svi_a2}
		\footnotesize\begin{array}{l}
			Z_{I2j}^k:= \left( \hspace{-1mm}\left(\widetilde{H}_{I}^k\right)^\top  \hspace{-1.5mm}- X_{I2j}^kY_{I2j}^{k}\right)^\top \hspace{-1.5mm} \left(X_{I2j}^k\right)^{-1} \hspace{-1mm}\left( \hspace{-1mm}\left(\widetilde{H}_{I}^k\right)^\top  \hspace{-1.5mm}- X_{I2j}^kY_{I2j}^{k}\right).
		\end{array}\normalsize
	\end{equation}
	
	\begin{figure*}
		\begin{equation}\label{eq:C1_S_proc}\tag{21}
			\begin{aligned}
				&2e_j^\top\widetilde{H}_{I}(A_{CL}\xi+B_{LC}z_s)-2\eta=-\left(1_{n+n_K}-\widetilde{H}_I\xi\right)^\top D_{I1j}\left(1_{n+n_K}+\widetilde{H}_I\xi\right) -\\ &\qquad\qquad\qquad\qquad\qquad\qquad-(1_{n+n_K}-H_{\mathcal{Z}}z_s)^\top D_{I2j}(1_{n+n_K}+H_{\mathcal{Z}}z_s) - \begin{bmatrix}
					\xi^\top & z_s^\top & 1
				\end{bmatrix}L_{C1j}\begin{bmatrix}
					\xi^\top & z_s^\top & 1
				\end{bmatrix}^\top.
			\end{aligned}
		\end{equation}
		\hrulefill
	\end{figure*}
	
	Now, since the employed optimizer satisfies all of the constraints which go into $\mathcal{P}^k\left(Q^k\right)$ (one of them being \eqref{eq:opt_g}), it follows that $M_{gj}^{k+1}-\mathrm{diag}\left(O, Z_{I2j}^k\right)\succ O$ and, recalling $X_{I2j}^k\succ O$, also that $M_{gj}^{k+1}\succ O$, for all $j\in\mathbb{N}_{[1,n+n_K]}$. Therefore, we are able to conclude that any optimizer of $\mathcal{P}^{k}\left(Q^{k}\right)$ satisfies constraint \eqref{eq:opt_g} for $\mathcal{P}^{k+1}\left(Q^{k+1}\right)$. The fact that the same statement holds for \eqref{eq:opt_h}, \eqref{eq:opt_i}, \eqref{eq:opt_k}, and \eqref{eq:opt_p} can be shown by applying, \emph{mutatis mutandis}, the same arguments as those employed in proving the statement concerning \eqref{eq:opt_g}.
	
	At this point, we have shown that all of the constraints which go into $\mathcal{P}^{k+1}\left(Q^{k+1}\right)$ are feasible for any optimizer of $\mathcal{P}^{k}\left(Q^{k}\right)$, where $k=0$. Thus, provided that the problem tackled in the initialization step of Algorithm~\ref{alg:syn_proc} is feasible, we can employ any of its optimizers to show that the first step of the algorithm's iterative phase involves solving a feasible problem. In order to obtain recursive feasibility, notice that all of the arguments employed to prove the transmission of feasibility from the initialization to the iteration's first step can also be applied for any iteration $k\in\mathbb{N}$. Consequently, a straightforward induction-based proof yields the recursive feasibility of Algorithm~\ref{alg:syn_proc} along with the fact that, for all optimizers produced by the latter's iterative phase, we have $X_{I1j}^k, X_{I2j}^k, X_{O1j}^k, X_{O2j}^k, X_{\mathcal{U}\ell}^k\succ O$ for all $k\in\mathbb{N}$ (recall the arguments made in the very beginning of the result's proof).

	We now conclude the proof of point a), by showing the guaranteed convergence of Algorithm~\ref{alg:syn_proc}. To do so, we employ the previously proven fact that an optimizer of $\mathcal{P}^k\left(Q^k\right)$, for an arbitrary iteration $k\in\mathbb{N}$, is feasible for $\mathcal{P}^{k+1}\left(Q^{k+1}\right)$. Then, by denoting via $J^k_{\mathrm{opt}}$ the optimal value obtained when solving $\mathcal{P}^k\left(Q^k\right)$, for all $k\in\mathbb{N}$, it is straightforward to notice that $J^k_{\mathrm{opt}}\geq J^{k+1}_{\mathrm{opt}}$ for all $k\in\mathbb{N}_{>0}$. Indeed, since any two consecutive problems share the same expression for the cost function (when $k\in\mathbb{N}_{>0}$), no optimizer of $\mathcal{P}^{k+1}\left(Q^{k+1}\right)$ can yield a greater cost value than an optimizer of $\mathcal{P}^{k}\left(Q^{k}\right)$, given that the latter's constituent matrices represents a feasible tuple for $\mathcal{P}^{k+1}\left(Q^{k+1}\right)$. Then, the sequence formed by the values of $J^k_{\mathrm{opt}}$ is monotonically non-increasing for all $k\in\mathbb{N}_{>0}$, and it is bounded from below by $0$, due to the cost function being the squared Frobenius norm of $G^k$. Therefore, we can employ the monotone convergence theorem to state that the sequence given by $J^k_{\mathrm{opt}}$ converges to some semi-positive value, which means that the stopping condition of Algorithm~\ref{alg:syn_proc} must be satisfied for some $k^\star\in\mathbb{N}_{>0}$.
	
	To prove point b), we first set $k=0$ and we then recall the invertibility of $Y_{I2j}^k$ along with the positive definiteness of $X_{I2j}^k$, which was previously shown in the proof of point~a) for all $j\in\mathbb{N}_{[1,n+n_K]}$. Then, \eqref{eq:opt_g} implies the fact that\vspace{-1mm}
	\begin{equation*}
		\left(\widetilde{H}_{I}^kY_{I2j}^k\right)^\top\hspace{-1.5mm}+\widetilde{H}_{I}^kY_{I2j}^k- \left(Y_{I2j}^k\right)^\top X_{I2j}^kY_{I2j}^k\succ O,\vspace{-1mm}
	\end{equation*}
	and, moreover, that\vspace{-1mm}
	\begin{equation*}
		\left(\widetilde{H}_{I}^kY_{I2j}^k\right)^\top\hspace{-1.5mm}+\widetilde{H}_{I}^kY_{I2j}^k \succ \left(Y_{I2j}^k\right)^\top X_{I2j}^kY_{I2j}^k\succ O.
	\end{equation*}
	Therefore, $\widetilde{H}_{I}^k$ must be invertible since, otherwise, any vector from its left nullspace could be used to refute the positive definiteness of $\left(\widetilde{H}_{I}^kY_{I2j}^k\right)^\top\hspace{-1.5mm}+\widetilde{H}_{I}^kY_{I2j}^k$. In addition to this, we also get that $Y_{I2j}^{k+1}=\left(X_{I2j}^{k}\right)^{-1}\left(\widetilde{H}_{I}^{k}\right)^\top$ is invertible which, coupled with the fact that $X_{I2j}^{k+1}\succ O$ (recalling the proof of point a), stated above), replicates the same two properties which were employed to construct the arguments made for the case $k=0$. Since the aforementioned arguments can be reused for any desired $k\in\mathbb{N}$, it becomes straightforward to prove by induction the fact that $Y_{I2j}^{k}$ and $\widetilde{H}_{I}^k$ are invertible, for all $j\in\mathbb{N}_{[1,n+n_K]}$ and all $k\in\mathbb{N}$. By exploiting, now, the invertibility of $Y_{I1j}^0$, $Y_{O1j}^0$ and $Y_{O2j}^0$ along with the positive definiteness of all $X_{I1j}^k$, $X_{O1j}^k$ and $X_{O2j}^k$, it is possible to employ \eqref{eq:opt_h}, \eqref{eq:opt_i} and \eqref{eq:opt_k} in an analogous manner to prove that $\widetilde{H}_{Ii}^k$, $H_O^k$ and $H_{Oi}^k$ \big(along with $Y_{I1j}^k$, $Y_{O1j}^k$ and $Y_{O2j}^k$\big) are all invertible, for each $k\in\mathbb{N}$ \big(and all $j\in\mathbb{N}_{[1,n+n_K]}$\big).
	
	Finally, we now address point c) of our result and, in order to do so, we first show how the set-theoretical conditions given in points O1)-O4) and C1)-C2) map onto the individual constraints of the optimization problem given in \eqref{eq:opt_a}-\eqref{eq:opt_s}.
	
	We begin with C1), which is satisfied whenever we have \begin{subequations}
		\begin{align}\label{eq:C1_ineq_a}
			&\phantom{-}\ 2e_j^\top\widetilde{H}_{I}(A_{CL}\xi+B_{LC}z_s)-2\eta\leq0,\\
			&-2e_j^\top\widetilde{H}_{I}(A_{CL}\xi+B_{LC}z_s)-2\eta\leq0,\label{eq:C1_ineq_b}
		\end{align}
	\end{subequations}
	for all $\xi\in\widetilde{\Omega}_I$, all $z_s\in\mathcal{Z}_s$ and all $j\in\mathbb{N}_{[1,n+n_K]}$. By exploiting the symmetry of both $\widetilde{\Omega}_I$ and $\mathcal{Z}_s$, note that satisfying any of the two inequalities in \eqref{eq:C1_ineq_a}-\eqref{eq:C1_ineq_b} guarantees that the other holds, and, thus, we proceed by considering only \eqref{eq:C1_ineq_a}. By doing so, we then rewrite the left-hand term from \eqref{eq:C1_ineq_a} as in \eqref{eq:C1_S_proc}, located at the top of this page, where $D_{I1j}$ and $D_{I2j}$ are diagonal and positive semidefinite matrices, and where we have defined\stepcounter{equation}
	\begin{equation}\label{eq:L_C1j_def}
		\hspace{-1mm}L_{C1j}:=\begin{bmatrix}
			\widetilde{H}_I^\top D_{I1j}\widetilde{H}_I & O & -A_{CL}^\top\widetilde{H}_I^\top e_j\\
			* & {H}_{\mathcal{Z}}^\top D_{I2j}{H}_{\mathcal{Z}} & -B_{CL}^\top\widetilde{H}_I^\top e_j\\
			* & * & r_{Ij}
		\end{bmatrix},\normalsize\hspace{-1mm}
	\end{equation}
	along with $
	r_{Ij}:=2\eta-1_{n+n_K}^\top D_{I1j}1_{n+n_K}^{}-1_{n_{\mathcal{Z}}}^\top D_{I2j}1_{n_{\mathcal{Z}}}^{},
	$
	for all $j\in\mathbb{N}_{[1,n+n_K]}$. Recalling, now, the expressions of the sets from \eqref{eq:set_H_reps}, the fact that $D_{I1j}$ and $D_{I2j}$ are diagonal and positive semidefinite, and that $\xi\in\widetilde{\Omega}_I$ along with $z_s\in\mathcal{Z}_s$, it is straightforward to notice the the term located on the left-hand side of \eqref{eq:C1_S_proc} is negative if $L_{C1j}\succ O$ or, equivalently (by Schur complement-based arguments), if $r_{Ij}>0$ and also
	\begin{multline*}
		\begin{bmatrix}
			\widetilde{H}_I^\top D_{I1j}\widetilde{H}_I & O\\
			O & {H}_{\mathcal{Z}}^\top D_{I2j}{H}_{\mathcal{Z}}
		\end{bmatrix}-\begin{bmatrix}
			A_{CL}^\top\\B_{CL}^\top
		\end{bmatrix}\widetilde{H}_I^\top e_j\frac{1}{r_{Ij}}e_j^\top\cdot\\
		\cdot \widetilde{H}_I\begin{bmatrix}
			A_{CL} & B_{CL}
		\end{bmatrix}\succ O,\,\forall\,j\in\mathbb{N}_{[1,n+n_K]}.
	\end{multline*}

	By applying Theorem III.1 in \cite{svi1}, the latter two conditions are equivalent to the existence of $X_{I1j}\succ O$ which satisfy
	\begin{equation}\label{eq:first_form}
		\left\{\begin{aligned}
			\begin{bmatrix}
				\widetilde{H}_I^\top D_{I1j}\widetilde{H}_I & O & A_{CL}^\top\\\vspace{-3mm}\\
				* & {H}_{\mathcal{Z}}^\top D_{I2j}{H}_{\mathcal{Z}} & B_{CL}^\top\\
				* & * & X_{I1j}
			\end{bmatrix}\succ O,\\
			\begin{bmatrix}
				X_{I1j}^{-1} & \widetilde{H}_I^\top e_j\\
				* & r_{Ij}
			\end{bmatrix}\succ O,\normalsize
		\end{aligned}\right.
	\end{equation}
	for all $j\in\mathbb{N}_{[1,n+n_K]}$. Moreover, by exploiting the invertibility of $\widetilde{H}_I$, recalling \eqref{eq:OI_lifted}, and by effecting a pair of nonsingular congruence transformations onto the inequalities in \eqref{eq:first_form}, we get that the latter are equivalent to
	\begin{subequations}
		\begin{align}\label{eq:second_form_a}
			\begin{bmatrix}
				D_{I1j} & O & \left(\widetilde{H}_I^\top\right)^{-1}A_{CL}^\top\\\vspace{-3mm}\\
				* & {H}_{\mathcal{Z}}^\top D_{I2j}{H}_{\mathcal{Z}} & B_{CL}^\top\\
				* & * & X_{I1j}
			\end{bmatrix}&\succ O,\\
			\begin{bmatrix}
				\left(\widetilde{H}_I^\top\right)^{-1}X_{I1j}^{-1}\,\widetilde{H}_I^{-1} &  e_j\\
				* & r_{Ij}
			\end{bmatrix}&\succ O.\label{eq:second_form_b}
		\end{align}
	\end{subequations}
	
	Denoting now $\widetilde{H}_{Ii}:=\widetilde{H}_I^{-1}$ and applying Theorem III.1 from \cite{svi1} to \eqref{eq:second_form_a}, we get that \eqref{eq:second_form_a}-\eqref{eq:second_form_b} are equivalent to the existence of $X_{I2j}\succ O$ which satisfy
	\begin{subequations}
		\begin{align}\label{eq:third_form_a}
			\begin{bmatrix}
				X_{I2j} & O & A_{CL}^\top\\\vspace{-3mm}\\
				* & {H}_{\mathcal{Z}}^\top D_{I2j}{H}_{\mathcal{Z}} & B_{CL}^\top\\
				* & * & X_{I1j}
			\end{bmatrix}&\succ O,\\
			\begin{bmatrix}
				D_{I1j} & \widetilde{H}_{Ii}^\top\\\vspace{-3mm}\\
				* & X_{I2j}^{-1}
			\end{bmatrix}&\succ O,\normalsize\label{eq:third_form_b}\\
			\begin{bmatrix}
				\widetilde{H}_{Ii}^\top\,X_{I1j}^{-1}\,\widetilde{H}_{Ii}^{} &  e_j\\
				* & r_{Ij}
			\end{bmatrix}&\succ O,\label{eq:third_form_c}
		\end{align}
	\end{subequations}
	for all $j\in\mathbb{N}_{[1,n+n_K]}$. Finally, in order to map the inequalities \eqref{eq:third_form_a}-\eqref{eq:third_form_c} onto the constraints from \eqref{eq:opt_a}-\eqref{eq:opt_s}, we effect a nonsingular congruence transformation on \eqref{eq:third_form_b} to obtain an equivalent condition for it, namely
	\begin{equation}\label{eq:third_form_equiv}
		\begin{bmatrix}
			D_{I1j} & I_{n+n_K}\\\vspace{-3mm}\\
			* & \widetilde{H}_IX_{I2j}^{-1}\widetilde{H}_I^\top
		\end{bmatrix}\succ O,
	\end{equation}
	and we employ the slack variable identity in \eqref{eq:third_form_equiv} and \eqref{eq:third_form_c}, just as in \eqref{eq:svi_a1}-\eqref{eq:svi_a2}, to get that \eqref{eq:third_form_a}-\eqref{eq:third_form_c} are satisfied if there exist $Y_{I1j}$ and $Y_{I2j}$ so that
	\begin{subequations}
		\begin{align}\label{eq:fourth_form_a}
			\begin{bmatrix}
				X_{I2j} & O & A_{CL}^\top\\\vspace{-3mm}\\
				* & {H}_{\mathcal{Z}}^\top D_{I2j}{H}_{\mathcal{Z}} & B_{CL}^\top\\
				* & * & X_{I1j}
			\end{bmatrix}&\succ O,\\
			\begin{bmatrix}
				D_{I1j} & I_{n+n_K}\\\vspace{-3mm}\\
				* & Y_{I2j}^{\top}\widetilde{H}_{I}^{\top}+\widetilde{H}_{I}^{}Y_{I2j}^{}-Y_{I2j}^{\top}X_{I2j}^{}Y_{I2j}^{}
			\end{bmatrix}&\succ O,\label{eq:fourth_form_b}\hspace{-1mm}\vspace{2mm}\\
			\begin{bmatrix}
				Y_{I1j}^{\top}\widetilde{H}_{Ii}^{}+\widetilde{H}_{Ii}^{\top}Y_{I1j}^{}-Y_{I1j}^{\top}X_{I1j}^{}Y_{I1j}^{} &  e_j\\
				* & r_{Ij}
			\end{bmatrix}&\succ O,\label{eq:fourth_form_c}\\
			r_{Ij}:=2\eta-1_{n+n_K}^\top D_{I1j}1_{n+n_K}^{}-1_{n_{\mathcal{Z}}}^\top D_{I2j}&1_{n_{\mathcal{Z}}}^{}\label{eq:fourth_form_d},
		\end{align}
	\end{subequations}
	for all $j\in\mathbb{N}_{[1,n+n_K]}$. It is now clear that if the cost function is $0$ at convergence, then $\widetilde{H}_{Ii}^{k^\star}=\left(\widetilde{H}_{I}^{k^\star}\right)^{-1}$ along with conditions
	
	\newpage\noindent
	\eqref{eq:opt_f}, \eqref{eq:opt_g}, \eqref{eq:opt_i} and \eqref{eq:opt_j} guarantee that \eqref{eq:fourth_form_a}-\eqref{eq:fourth_form_d} hold and, so, the matrices defined in \eqref{eq:L_C1j_def} are positive definite, for all $j\in\mathbb{N}_{[1,n+n_K]}$. Finally, \eqref{eq:opt_b} ensures that $D_{I1j}^{k^\star}$ and $D_{I2j}^{k^\star}$ are diagonal and positive semidefinite, for each $j\in\mathbb{N}_{[1,n+n_K]}$, satisfying the conditions in \eqref{eq:C1_ineq_a}-\eqref{eq:C1_ineq_b} and making C1) hold.
	
	By applying now the same arguments as those used to show that C1) holds, it is straightforward to prove that 
	\eqref{eq:opt_c} and \eqref{eq:opt_m} along with \eqref{eq:opt_d} and \eqref{eq:opt_n} are employed to ensure that $\widetilde{\Omega}_I\subseteq\left((\varepsilon_s-\varepsilon_m-\varepsilon_p)\,\mathcal{S}\right)\hspace{-0.5mm}\times\hspace{-0.5mm}\mathbb{R}^{n_K}$ and that $\Omega_O\subseteq\mathcal{X}\times\mathbb{R}^{n_K}$, respectively, which guarantee that C2) along with O3) are satisfied. Similarly, it follows by tedious algebraic manipulations that \eqref{eq:opt_b},\eqref{eq:opt_f}, \eqref{eq:opt_h}, \eqref{eq:opt_k} and \eqref{eq:opt_l} are used in the optimization problem to ensure that O2) holds, while \eqref{eq:opt_e}, \eqref{eq:opt_o} and \eqref{eq:opt_p} are tasked with ensuring that O4) is satisfied. Finally, and perhaps most straightforwardly, the conditions stated in \eqref{eq:opt_q} ensure that all the vertices of the polyhedron $\mathcal{S}_c\times\{0_{n_K}\}$ belong to the set $\Omega_I$. Since the latter is convex and the points which make up the former are convex combinations of its vertices, the inclusion which makes up O1) is automatically satisfied. \qed


	\bibliographystyle{IEEEtran}
	\bibliography{manuscript}\end{document}